\documentclass[AMA,STIX1COL]{WileyNJD-v2}
\articletype{Research Article}%
\usepackage{arydshln}
\usepackage{makecell}
\usepackage{multirow}

\newcommand{\negphantom}[1]{\settowidth{\dimen0}{#1}\hspace*{-\dimen0}}
\newcommand{\trimmedref}[1]{\ref{#1}\negphantom{\enspace}}
\usepackage{url}
\usepackage{mathtools}
\DeclareMathOperator{\ourtimes}{\ast}
\DeclareMathOperator{\trunc}{\mathrm{trunc}}

\DeclarePairedDelimiterXPP\ourfloor[1]{}{\lfloor} {\rfloor}{}{#1}
\DeclarePairedDelimiterXPP\ourceiling[1]{}{\lceil} {\rceil}{}{#1}

\DeclareMathOperator*{\lsb}{\textrm{lsb}}

\newcommand{\reciprocal}{\ensuremath{c}}
\DeclareMathOperator*{\signedmod}{\bmod}
\DeclareMathOperator*{\ourdiv}{\mathrm{div}}
\DeclareMathOperator*{\signeddiv}{\ourdiv}
\usepackage{graphicx}
\usepackage{underscore}
\graphicspath{{.}}
\DeclareGraphicsExtensions{.png,.pdf,.jpg}
\usepackage{float}
\usepackage{url}
\usepackage{breakurl}
\usepackage{booktabs}
\usepackage{amssymb,amsmath,amsthm}
\usepackage{xcolor}
\hypersetup{
    colorlinks,
    linkcolor={red!50!black},
    citecolor={blue!50!black},
    urlcolor={blue!80!black}
}

\usepackage{todonotes}
\usepackage{csquotes}
\usepackage{listings}
\lstdefinelanguage
   [x64]{Assembler}     
   [x86masm]{Assembler} 
{morekeywords={CDQE,CQO,CMPSQ,CMPXCHG16B,JRCXZ,LODSQ,MOVSXD, %
                  POPFQ,PUSHFQ,SCASQ,STOSQ,IRETQ,RDTSCP,SWAPGS, movabs, %
                  rax,rdx,rcx,rbx,rsi,rdi,rsp,rbp, %
                  r8,r8d,r8w,r8b,r9,r9d,r9w,r9b}} 
\lstdefinestyle{customc}{%
  belowcaptionskip=1\baselineskip,
  breaklines=true,
  xleftmargin=\parindent,
  language=C,
  showstringspaces=false,
  basicstyle=\ttfamily,
  keywordstyle=\bfseries\color{green!30!black},
  numberstyle=\tiny,
  commentstyle=\color{purple!30!black},
  identifierstyle=\bfseries\color{black},
  stringstyle=\color{orange},  morekeywords={uint64_t,uint32_t,int64_t,__int128_t,int32_t,__m256i,__m128i,UINT64_C},
}
\usepackage{subfloat,subfig}
\usepackage{siunitx}

\usepackage{amsthm}
\newtheorem*{example}{Example}
\received{}
\revised{}
\accepted{}

\begin{document}
\newcommand{\ourtitle}{Faster Remainder by Direct Computation}
\newcommand{\oursubtitle}{Applications to Compilers and Software Libraries}

\title{\ourtitle{}\\\subtitlefont \oursubtitle }

\author[1]{Daniel Lemire*}

\author[2]{Owen Kaser}

\author[3]{Nathan Kurz}

\authormark{D. LEMIRE, O. KASER, N. KURZ}

\address[1]{\orgdiv{TELUQ}, \orgname{Universit\'e du Qu\'ebec}, \orgaddress{\state{Quebec}, \country{Canada}}}

\address[2]{\orgdiv{Computer Science Department}, \orgname{UNB Saint John}, \orgaddress{\state{New Brunswick}, \country{Canada}}}

\address[3]{\orgaddress{\state{Vermont}, \country{USA}}}

\corres{*Daniel Lemire, 5800 Saint-Denis, Office 1105, Montreal, Quebec, H2S 3L5 Canada. \email{lemire@gmail.com}}

\abstract[Summary]{%
On common processors, integer multiplication is many times faster than integer  division.
Dividing a numerator $n$ by a divisor $d$ is mathematically equivalent to multiplication by the inverse of the divisor ($n / d = n \ourtimes 1/d$).
If the divisor is known in advance---or if repeated integer divisions will be performed with the same divisor---it
can be beneficial to substitute a less costly multiplication for an expensive division.

 Currently, the remainder of the division by a constant is computed from the quotient by a multiplication and a subtraction.
But if just the remainder is desired and the quotient is unneeded, this may be suboptimal.
We present a generally applicable algorithm to compute the remainder more directly. Specifically, we use the fractional portion of the product of the numerator and the inverse of the divisor. On this basis,
we also present a new, simpler divisibility algorithm to detect nonzero remainders.

We also derive new tight bounds
on the precision required when representing the inverse of the divisor.
Furthermore, we present simple C implementations  that beat
the optimized code produced by state-of-art C compilers on recent x64  processors (e.g., Intel Skylake and AMD Ryzen), sometimes by more than 25\%.
On all tested platforms including 64-bit ARM and POWER8, our divisibility-test functions are  faster than  state-of-the-art Granlund-Montgomery divisibility-test functions, sometimes by more than 50\%.
}


\jnlcitation{\cname{%
\author{D. Lemire},
\author{O. Kaser}, and
\author{N. Kurz}},
\ctitle{\ourtitle{}}, \cvol{2018}.}

\maketitle

\thispagestyle{empty}

\section{Introduction}

Integer division often refers to two closely related concepts, the
actual division and the modulus. 
Given an integer numerator $n$ and a non-zero integer divisor $d$, the integer division,
written $\ourdiv$,  gives  the integer quotient
($n \ourdiv d = q$). The modulus, written $\bmod$, gives the integer remainder ($n \bmod d = r$).   Given an integer numerator $n$ and an integer divisor $d$, the quotient ($q$) and the remainder ($r$) are always integers even when the fraction $n/d$ is not an integer.   It always holds that the quotient multiplied by the divisor plus the remainder gives back the numerator: $n = q \ourtimes d + r$.

Depending on the context, `integer division' might refer solely to the computation of the quotient, but might also refer to the computation of both the integer quotient and the remainder.  The integer division instructions on x64 processors compute both the quotient and the remainder.\footnote{We use \emph{x64} to refer to the commodity Intel and AMD processors supporting the 64-bit version of the x86 instruction set. It is also known as x86-64, x86\_64, AMD64 and Intel~64.} In most programming languages, they are distinct operations: the C programming language uses \texttt{/} for division ($\ourdiv$) and \texttt{\%} for modulo ($\bmod$).

\newcommand{\stones}{items}
Let us work through a simple example to illustrate how we can replace an integer division by a multiplication.  Assume we have a pile of 23~\stones, and we want to know how many piles of 4~\stones\ we can divide it into and how many will be left over
($n=23$, $d=4$; find $q$ and $r$).  Working in base 10, we can
 calculate the quotient $23 \ourdiv 4 = 5$ and the remainder $23 \bmod 4 =
3$, which means that there will be 5~complete piles of 4~\stones\ with 3~\stones\
left over ($23 = 5 \ourtimes 4 + 3$).  

If for some reason, we do not have a runtime
integer division operator (or if it is too expensive for our purpose), we can instead
precompute the multiplicative inverse of 4 once ($\reciprocal=1/d=1/4=0.25$) and then
calculate the same result
using a multiplication ($23 \ourtimes 0.25 = 5.75$).  The quotient is the integer portion of the product to the left
of the decimal point ($q = \ourfloor*{5.75} = 5$), and the remainder can be
obtained by multiplying the fractional portion $f=0.75$ of the product
by the divisor $d$: $r = f \ourtimes d = 0.75 \ourtimes 4 = 3$.

The binary registers in our computers do not have a built-in concept of a fractional portion, but we can adopt a fixed-point
convention.   Assume we have chosen a convention
where $1/d$ has 5~bits of whole integer value and 3~bits of `fraction'.
The numerator 23 and divisor 4 would still be represented as  standard 8-bit binary values
(00010111 and 00000100, respectively),
but $1/d$ would be 00000.010.   From the processor's
viewpoint, the rules for arithmetic are still the same as if we did
not have a binary point---it is
only our interpretation of the units that has changed.   Thus we can
use the standard (fast) instructions for multiplication ($00010111
\ourtimes 00000010 = 00101110$) and then mentally put the `binary point' in the correct
position, which in this case is 3 from the right: 00101.110.
The quotient $q$ is the integer portion (leftmost 5 bits) of this result: 00101 in
binary ($q=5$ in decimal).
 In effect, we can compute the quotient $q$ with a multiplication (to get 00101.110) followed
by a right shift (by three bits, to get 000101).
To find the remainder, we can multiply the fractional portion (rightmost 3~bits)
of the result by the divisor: $00000.110 \ourtimes 00000100 = 00011.000$ ($r = 3$
in decimal).
To quickly check whether a number is divisible by 4 ($n \mod 4 = 0$?) without computing the remainder it suffices to check whether the fractional portion of the product is zero.

But what if instead of dividing by 4, we wanted to divide by 6?  While
the reciprocal of 4 can be represented exactly with two digits
of fixed-point fraction in both binary and decimal,  $1/6$ cannot be
exactly represented in either.  As a decimal fraction, $1/6$ is equal to
the repeating fraction  0.1666\ldots (with a repeating 6), and in binary
it is 0.0010101\ldots (with a repeating 01).
Can the same technique
work if we have a sufficiently close approximation to the reciprocal for any divisor, using enough fractional bits?  Yes!  

For example, consider a convention
where the approximate reciprocal $\reciprocal$ has 8 bits, all of which are fractional. We can use the value 0.00101011 as our approximate reciprocal of 6. To divide $23$ by 6, we can multiply the numerator (10111 in binary) by the approximate reciprocal: $n \ourtimes \reciprocal = 00010111 \ourtimes 0.00101011 = 11.11011101$. As before, the decimal point is merely a convention, the computer need only 
multiply fixed-bit integers. From the product, the quotient of the  division is 11 in binary ($q=3$ in decimal); and indeed $23\ourdiv 6=3$.  To get the remainder, we multiply the fractional portion of the product by the divisor ($f \ourtimes
d = 0.11011101 \ourtimes 00000110 = 101.00101110$),
and then right shift
by 8~bits, to get 101 in binary ($r=5$ in decimal). See Table~\ref{tab:example6} for other examples.

\begin{table}
\caption{\label{tab:example6} Division by 6 ($d=110$ in binary) using a multiplication by the approximate reciprocal ($\reciprocal=0.00101011$ in binary).
The numerator $n$ is an $N$-bit value,  with $N=6$. The approximate reciprocal uses $F=8$~fractional bits. The integer portion of the product ($N$ bits) gives the quotient. Multiplying the fractional portion of the product ($F$ bits) by the divisor ($N$ bits) and keeping only the integer portion ($N$ bits), we get the remainder (two last columns). The integer portion in bold (column 2) is equal to the quotient (column 3) in binary. The integer portion in bold (column 4) is equal to the remainder (column 5) in binary.
}
\centering
\begin{minipage}{\textwidth}
\centering
\begin{tabular}{clclc}\toprule
$n$     & numerator times the approx.\ reciprocal ($n\ourtimes \reciprocal$)                             & quotient & fractional portion $\ourtimes$ divisor  ($f \ourtimes d$) & remainder  \\ 
    & $N$ bits $\ourtimes$ $F$ bits  $\to$ $N+F$ bits                             & $N$ bits & $F$ bits $\ourtimes$ $N$ bits $\to$ $N+F$ bits & $N$ bits \\ \midrule
0 &   $ 000000 \ourtimes 0.00101011 = \textbf{000000}.00000000$ & 0  &  $ 0.00000000 \ourtimes 000110 = \textbf{000000}.00000000$ & 0 \\
1 &   $ 000001 \ourtimes 0.00101011 = \textbf{000000}.00101011$ & 0  &  $ 0.00101011 \ourtimes 000110 = \textbf{000001}.00000010$ & 1 \\
2 &   $ 000010 \ourtimes 0.00101011 = \textbf{000000}.01010110$ & 0  &  $ 0.01010110 \ourtimes 000110 = \textbf{000010}.00000100$ & 2 \\
3 &   $ 000011 \ourtimes 0.00101011 = \textbf{000000}.10000001$ & 0  &  $ 0.10000001 \ourtimes 000110 = \textbf{000011}.00000110$ & 3 \\
4 &   $ 000100 \ourtimes 0.00101011 = \textbf{000000}.10101100$ & 0  &  $ 0.10101100 \ourtimes 000110 = \textbf{000100}.00001000$ & 4 \\
5 &   $ 000101 \ourtimes 0.00101011 = \textbf{000000}.11010111$ & 0  &  $ 0.11010111 \ourtimes 000110 = \textbf{000101}.00001010$ & 5 \\
6 &   $ 000110 \ourtimes 0.00101011 = \textbf{000001}.00000010$ & 1  &  $ 0.00000010 \ourtimes 000110 = \textbf{000000}.00001100$ & 0 \\
$\vdots$ &  $\vdots$ & $\vdots$  &  $\vdots$ & $\vdots$ \\
17 &   $ 010001 \ourtimes 0.00101011 = \textbf{000010}.11011011$ & 2  &  $ 0.11011011 \ourtimes 000110 = \textbf{000101}.00100010$ & 5 \\
18 &   $ 010010 \ourtimes 0.00101011 = \textbf{000011}.00000110$ & 3  &  $ 0.00000110 \ourtimes 000110 = \textbf{000000}.00100100$ & 0 \\
19 &   $ 010011 \ourtimes 0.00101011 = \textbf{000011}.00110001$ & 3  &  $ 0.00110001 \ourtimes 000110 = \textbf{000001}.00100110$ & 1 \\
20 &   $ 010100 \ourtimes 0.00101011 = \textbf{000011}.01011100$ & 3  &  $ 0.01011100 \ourtimes 000110 = \textbf{000010}.00101000$ & 2 \\
21 &   $ 010101 \ourtimes 0.00101011 = \textbf{000011}.10000111$ & 3  &  $ 0.10000111 \ourtimes 000110 = \textbf{000011}.00101010$ & 3 \\
22 &   $ 010110 \ourtimes 0.00101011 = \textbf{000011}.10110010$ & 3  &  $ 0.10110010 \ourtimes 000110 = \textbf{000100}.00101100$ & 4 \\
23 &   $ 010111 \ourtimes 0.00101011 = \textbf{000011}.11011101$ & 3  &  $ 0.11011101 \ourtimes 000110 = \textbf{000101}.00101110$ & 5 \\
24 &   $ 011000 \ourtimes 0.00101011 = \textbf{000100}.00001000$ & 4  &  $ 0.00001000 \ourtimes 000110 = \textbf{000000}.00110000$ & 0 \\
$\vdots$ &  $\vdots$ & $\vdots$  &  $\vdots$ & $\vdots$ \\
63 &   $ 111111 \ourtimes 0.00101011 = \textbf{001010}.10010101$ & 10  &  $ 0.10010101 \ourtimes 000110 = \textbf{000011}.01111110$ & 3 \\

\bottomrule
\end{tabular}
\end{minipage}
\end{table}

While the use of the approximate reciprocal $\reciprocal$ prevents us from confirming divisibility by 6 by checking whether the fractional portion is exactly zero, we can still quickly determine whether a number is divisible by 6 ($n \bmod 6 = 0$?) by checking whether the fractional portion is less than the approximate reciprocal ($f < c$?).
Indeed, if $n = q \ourtimes d + r$ then the fractional portion of the product of $n$ with the approximate reciprocal 
should be close to $r/d$: it makes intuitive sense that comparing $r/d$ with $\reciprocal \approx 1/d$ determines whether the remainder $r$ is zero.
For example, consider $n=42$ (101010 in binary). We have that our numerator times the approximate reciprocal of 6 is $101010 \ourtimes 0.00101011 = 111.00001110$. We see that the quotient is 111 in binary ($q=7$ in decimal), while the fractional portion $f$ is smaller than the approximate reciprocal $\reciprocal$ ($0.00001110 < 0.00101011$), indicating that 42 is a multiple of 6.

In our example with 6 as the divisor, we used 8~fractional bits.  The more fractional bits we use, the
larger the numerator we can handle.
An insufficiency of fractional bits can lead to incorrect results when $n$ grows.
For instance, with $n=131$ ($10000011$ in binary) and only 8~fractional bits, $n$ times the approximate  reciprocal of
6  is  $10000011 \ourtimes 0.00101011 = 10110.00000001$, or 22 in decimal.
Yet in actuality, $131 \ourdiv 6 = 21$; using an 8~bit approximation of the reciprocal was inadequate.

How close does the approximation need to be?---that is, what is the minimum number of fractional bits needed for the approximate  reciprocal $\reciprocal$ such that the  remainder is exactly correct
for all numerators?
We derive the answer in \S~\ref{sec:remainderdirectly}.

The scenario we describe with an expensive division applies to current processors. Indeed, integer division instructions on recent x64 processors have a latency of 26~cycles for 32-bit registers and at least 35~cycles for 64-bit registers~\cite{fog2016instruction}.
We find similar latencies in the popular ARM processors~\cite{arma57}. Thus, most optimizing compilers replace
integer divisions by constants $d$ that are known at compile time with the equivalent of a multiplication by a scaled approximate reciprocal $\reciprocal$ followed by a shift.
To compute the remainder by a constant ($n \bmod d$), an optimizing compiler might first compute the quotient $n \ourdiv d$ as a multiplication by $\reciprocal$ followed by a logical shift by $F$~bits  $(  \reciprocal \ourtimes n ) \ourdiv 2^F$, and then use the fact that the remainder can be derived using a multiplication and a subtraction as $n \bmod d = n - (n \ourdiv d) \ourtimes d$.

Current optimizing compilers discard the fractional portion of the multiplication ($ (   \reciprocal \ourtimes n ) \bmod 2^F$).
Yet using the fractional bits to compute the remainder or test the divisibility in software has merit.
It can be faster (e.g., by more than 25\%) to  compute the remainder using the fractional bits compared  to the code produced for some processors (e.g., x64 processors)
by a state-of-the-art optimizing compiler.

\section{Related Work}

Jacobsohn\cite{jacobsohn1973combinatoric} derives an integer division method for unsigned integers by constant divisors.
After observing that any integer divisor can be written as an odd integer multiplied by a power of two, he focuses on the division by an odd divisor. He finds that we can divide by an odd integer  by multiplying by a fractional inverse, followed by some rounding. He presents an implementation solely with full adders, suitable for hardware implementations. He observes that we can get the remainder from the fractional portion with rounding, but he does not specify the necessary rounding or the number of bits needed.

In the special case where we divide by 10, Vowels\cite{Vowels:1992:D} describes the computation of both the quotient and remainder. In contrast with other related work, Vowels presents the computation of the remainder directly from the fractional portion. Multiplications are replaced by additions and shifts. He does not extend the work beyond the division by~10.

Granlund and Montgomery\cite{Granlund:1994:DII:773473.178249} present the first general-purpose algorithms to divide unsigned and signed integers by constants. Their approach relies on a multiplication followed by a division by a power of two which is implemented as an logical shift ($(\ourceiling*{2^F / d } \ourtimes n ) \ourdiv 2^F$).
 They implemented their approach  in the GNU GCC compiler, where it can still be found today (e.g., up to GCC version 7). Given any non-zero 32-bit divisor known at compile time, the optimizing compiler can (and usually will) replace the division by a multiplication followed by a shift.
Following Cavagnino and Werbrouck\cite{Cavagnino:2008:EAI:1388169.1388172}, Warren\cite{warr:hackers-delight-2e} finds that Granlund and Montgomery choose a slightly suboptimal
number of fractional bits for some divisors. Warren's slightly better approach is found in LLVM's Clang compiler. See Algorithm~\ref{algo:gmwalgo}.


\begin{algorithm}
\begin{algorithmic}[1]
\State \textbf{Require}: fixed integer divisor $d \in [1,2^{N})$
\State \textbf{Require}: runtime integer numerator $n \in [0,2^{N})$
\State \textbf{Compute}:  
the integer $n \ourdiv d$
\If{$d$ is a power of two ($d=2^K$)}
\State return $n\ourdiv 2^K$ \hfill \Comment{Implemented with a bitwise shift}
\ElsIf{for $L=\ourfloor*{\log_2(d)}$ and $\reciprocal=\ourceiling*{2^{N+L}/d}$, we have $\reciprocal \ourtimes (2^N- (2^N \bmod d) -1) < (2^N \ourdiv d) 2^{N+L} $ }
\State return $(\reciprocal \ourtimes n) \ourdiv 2^{N+L}$  \hfill\Comment{$\reciprocal\in[0,2^N)$}
\ElsIf{$d= 2^K d'$ for $K>0$}
\State let $L=\ourceiling*{\log_2(d')}$ and $\reciprocal=\ourceiling*{2^{N-K+L}/d'}$  \hfill\Comment{$\reciprocal\in[0,2^N)$}
\State return $ \reciprocal \ourtimes (n \ourdiv 2^K)   \ourdiv 2^{N-K+L}$
\Else
\State let  $L = \ourceiling*{\log_2 d }$ and  $\reciprocal=\ourceiling*{2^{N+L}/d }$ \hfill  \Comment{$\reciprocal>2^N$}
\State let $\reciprocal'$ be such that  $\reciprocal= 2^N + \reciprocal'$ \hfill \Comment{$\reciprocal'\in [0, 2^N)$}
\State return $(\ourfloor*{\reciprocal' \ourtimes n / 2^N } + ((n - \ourfloor*{\reciprocal' \ourtimes n / 2^N }) \ourdiv 2))\ourdiv 2^{L-1}$
\EndIf
\end{algorithmic}
\caption{Granlund-Montgomery-Warren\cite{Granlund:1994:DII:773473.178249,warr:hackers-delight-2e} division algorithm  by a constant using unsigned integers.  \label{algo:gmwalgo}}
\end{algorithm}

Following Artzy et al.\cite{Artzy:1976:FDT:359997.360013}, Granlund and Montgomery\cite{Granlund:1994:DII:773473.178249} describe how to check that an integer is a multiple of a constant divisor more cheaply than by the computation of the remainder. Yet, to our knowledge, no compiler uses this optimization. Instead, all compilers that we tested compute the  remainder $r$ by a constant  using the formula $r= n - q\ourtimes d$ and then compare against zero.  That is, they use a constant to compute the quotient, multiply the quotient by the original divisor, subtract from the original numerator, and only finally check whether the remainder is zero. 

In support of this approach, Granlund and Montgomery\cite{Granlund:1994:DII:773473.178249} state that \emph{the remainder, if desired, can be computed by an additional multiplication and subtraction}. Warren\cite{warr:hackers-delight-2e} covers the computation of the remainder without computing the quotient, but only for divisors that are a power of two $2^K$, or for small divisors that are nearly a power of two ($2^K +1$, $2^K -1$).

In software, to our knowledge, no authors except Jacobsohn\cite{jacobsohn1973combinatoric} and Vowels\cite{Vowels:1992:D} described using the fractional portion to compute the remainder or test the divisibility, and neither of these considered the general case. 
In contrast, the computation of the remainder directly, without first computing the quotient, has received some attention in the hardware and circuit literature\cite{331617,Doran1995,7933010}.
 Moreover, many researchers\cite{Rutten:2010:EFP:1731022.1731026,moller2011improved}
 consider the computation of the remainder of unsigned division by a small divisor to be useful when working with big integers (e.g., in cryptography).

\section{Computing the Remainder Directly}
\label{sec:remainderdirectly}

Instead of discarding the least significant bits resulting from the multiplication in the Granlund-Montgomery-Warren algorithm, we can use them to compute the remainder without ever computing the quotient.
We formalize this observation by Theorem~\ref{theorem:crazyass}, which we believe to be a novel mathematical result.

 In general, we expect that it takes at least $N$~\emph{fractional bits} for the approximate reciprocal $\reciprocal$ to provide exact computation of the remainder for all non-negative numerators less than $2^N$. Let us say we use $F=N+L$ fractional bits for some non-negative integer value $L$ to be determined. We want to pick $L$ so that approximate reciprocal $\reciprocal =\ourceiling*{2^F/d} = \ourceiling*{2^{N+L}/d}$ allows exact computation of the remainder as $r = \left ( \left ( \left ( \reciprocal \ourtimes n  \right ) \bmod 2^F \right )\ourtimes d \right ) \ourdiv 2^F$ where $F=N+L$.

We illustrate our notation using the division of $63$ by $6$ as in the last row of Table~\trimmedref{tab:example6}. We have that $63$ is 111111 in binary and that   the approximate reciprocal $\reciprocal$ of $6$ is $\ourceiling*{2^8/6}=00101011$ in binary.  We can compute the quotient as the integer part of the product of the reciprocal by 
\begin{align*} n \ourtimes \reciprocal = \overbrace{111111}^{N=6} \ourtimes 0.\underbrace{\overbrace{001010}^{N=6}\overbrace{11}^{L=2}}_{F=N+L=8}
= \overbrace{0001010}^{\mathrm{quotient:~}N\mathrm{~bits}}.\overbrace{10010101}^{( \reciprocal \ourtimes n   ) \bmod 2^F  }. \end{align*} 
Taking the $F$-bit fractional portion ($( \reciprocal \ourtimes n   ) \bmod 2^F$), and multiplying it by the divisor $6$ (110 in binary), we get the remainder as the integer portion of the result:
\begin{align*} (( \reciprocal \ourtimes n   ) \bmod 2^F  )\ourtimes d  =
\overbrace{10010101}^{( \reciprocal \ourtimes n   ) \bmod 2^F  } \ourtimes \overbrace{000110}^{d: N\mathrm{~bits}}=\overbrace{0000011}^{\mathrm{remainder:~}N\mathrm{~bits}}.\cdots.
\end{align*} 
The fractional portion $(\reciprocal \ourtimes n) \bmod 2^{F}$ given by 10010101 is relatively  close to the product of the reciprocal by the remainder ($00101011 \ourtimes 11$) given by 10000001: as we shall see, this is not an accident.

Indeed, we begin by showing that the $F=N+L$~least significant bits of the product ($(\reciprocal \ourtimes n) \bmod 2^{N+L}$) are approximately equal to the scaled approximate reciprocal $\reciprocal$ times the remainder we seek  ($\reciprocal \ourtimes (n\bmod d)$), in a way made precise by Lemma~\ref{lemma:interval}. Intuitively, this intermediate result is useful because we only need to multiply this product by $d$ and divide by $2^F$ to \emph{cancel out} $\reciprocal$ (since $\reciprocal \ourtimes d  \approx 2^F$) and get the remainder $n\bmod d$.

\begin{lemma}\label{lemma:interval} Given $d\in [1,2^N)$, and non-negative integers $\reciprocal, L $ such that \begin{eqnarray*}2^{N+L} \leq \reciprocal \ourtimes d \leq 2^{N+L} +2^L\end{eqnarray*} then
 \begin{eqnarray*}\reciprocal \ourtimes (n\bmod d) \leq (\reciprocal \ourtimes n) \bmod 2^{N+L} \leq \reciprocal \ourtimes (n\bmod d) + 2^{L} (n \ourdiv d) < 2^{N+L}\end{eqnarray*} for all $n \in [0,2^N)$.
\end{lemma}
\begin{proof}
We can write  $n$ uniquely as $n = q \ourtimes d + r$ for some integers $q$ and $r$ where $q\geq 0$ and $r \in [0,d)$. We assume that $2^{N+L} \leq \reciprocal \ourtimes d \leq 2^{N+L} +2^L$.

We begin by showing that
$\reciprocal\ourtimes r + 2^{L} q < 2^{N+L}$.
Because $\reciprocal\ourtimes d \leq 2^{N+L} +2^L$, we have that
\begin{eqnarray*}
\reciprocal \ourtimes r + 2^{L} q & \leq & \frac{2^{N+L}}{d} r + \frac{2^L}{d} r + 2^{L} q \\
 & = & \frac{2^L}{d}\left ( 2^N r +   r + d \ourtimes q \right ) \\
& = & \frac{2^L}{d}\left(n + 2^N r \right ).\end{eqnarray*} Because $n<2^N$ and $r<d$, we have that $n + 2^N r < 2^N d$ which shows that
\begin{equation}
\reciprocal \ourtimes r + 2^{L} q < 2^{N+L}. \label{eq:firstinequality}
\end{equation}

We can rewrite our assumption $2^{N+L}\leq \reciprocal \ourtimes d \leq 2^{N+L} +2^L$ as $0\leq \reciprocal \ourtimes d - 2^{N+L} \leq  2^L$. Multiplying throughout by the non-negative integer $q$, we get
\begin{eqnarray*}
0 \leq \reciprocal \ourtimes d \ourtimes q  - 2^{N+L} q \leq  2^L q.
\end{eqnarray*}
After adding $\reciprocal \ourtimes r$ throughout, we get
\begin{eqnarray*}
\reciprocal \ourtimes r \leq \reciprocal \ourtimes n - 2^{N+L} q \leq  2^L q + \reciprocal \ourtimes r
\end{eqnarray*}
where we used the fact that $\reciprocal \ourtimes d \ourtimes q  + \reciprocal \ourtimes r = \reciprocal \ourtimes n$.
So we have that $\reciprocal \ourtimes n - 2^{N+L} q \in [\reciprocal \ourtimes r, 2^L q + \reciprocal \ourtimes r]$.
We already showed (see Equation~\ref{eq:firstinequality})
that $2^L q + \reciprocal \ourtimes r$  is less than $2^{N+L}$ so that $ \reciprocal \ourtimes n - 2^{N+L} q \in [0,2^{N+L})$.
Thus we have that  $\reciprocal \ourtimes n - 2^{N+L} q = ( \reciprocal \ourtimes n )\bmod 2^{N+L}$
because (in general and by definition) if $ p - k Q \in [0, y)$ for some $y\leq Q$,  then $p \bmod Q = p - k Q $.
Hence, we have that $( \reciprocal \ourtimes n )\bmod 2^{N+L} \in [\reciprocal \ourtimes r, 2^L q + \reciprocal \ourtimes r]$.
This completes the proof.
\end{proof}

Lemma~\ref{lemma:interval} tells us that $(\reciprocal \ourtimes n) \bmod 2^{N+L}$ is close to $\reciprocal \ourtimes (n\bmod d)$ when $\reciprocal$ is close to $2^{N+L}/d$. Thus it should make intuitive sense that $(\reciprocal \ourtimes n) \bmod 2^{N+L}$ multiplied by $d/2^{N+L}$ should give us $n\bmod d$. The following theorem makes the result precise.

\begin{theorem}\label{theorem:crazyass} Given $d\in [1,2^N)$, and non-negative integers $\reciprocal, L $ such that
\begin{eqnarray*}\frac{1}{d} \leq \frac{\reciprocal}{2^{N+L}} \leq \frac{1}{d}
+ \frac{1/d}{2^{N}}\end{eqnarray*} 
then
\begin{eqnarray*}n \bmod d = \left (\left((\reciprocal \ourtimes n) \bmod 2^{N+L}\right) \ourtimes d \right) \ourdiv 2^{N+L}\end{eqnarray*} for all $n \in [0,2^N)$.
\end{theorem}
\begin{proof}
We can write $n$ uniquely  as $n = q \ourtimes d + r$ where $q\geq 0$ and $r \in [0,d)$.
By Lemma~\ref{lemma:interval},
we have that
 $\reciprocal \ourtimes n \bmod 2^{N+L} \in [\reciprocal \ourtimes r, \reciprocal \ourtimes r +2^L q]$ for all $n \in [0,2^N)$.

We want to show that if we multiply any value in $[\reciprocal \ourtimes r, \reciprocal \ourtimes r +2^L q]$ by $d$ and divide it by $2^{N+L}$, then we get $r$.
That is, if $y\in [\reciprocal \ourtimes r, \reciprocal \ourtimes r +2^L q]$, then $d\ourtimes y  \in [2^{N+L} r, 2^{N+L} (r+1))$. We can check this inclusion using two inequalities:
\begin{itemize}
\item ($d\ourtimes y \geq 2^{N+L} r$) It is enough to show that
$\reciprocal \ourtimes d \ourtimes r \geq 2^{N+L} r$ which follows since $\reciprocal \ourtimes d \geq 2^{N+L}$ by one of our assumptions.
\item ($d\ourtimes y < 2^{N+L} (r+1)$) It is enough to show that
$d \ourtimes ( \reciprocal \ourtimes r +2^L q) < 2^{N+L} (r+1)$.
Using the assumption that $\reciprocal \ourtimes d \leq 2^{N+L}+ 2^L$, we have
that $d \ourtimes ( \reciprocal \ourtimes r +2^L q) \leq 2^{N+L} r + 2^L r + 2^L  d \ourtimes q
= 2^{N+L} r + 2^L n$. Since $n < 2^N$, we finally have
$d \ourtimes ( \reciprocal \ourtimes r +2^L q)<2^{N+L} (r+1)$ as required.
\end{itemize}
This concludes the proof.
\end{proof}

Consider the constraint $2^{N+L} \leq \reciprocal \ourtimes d \leq 2^{N+L} +2^L$ given by Theorem~\ref{theorem:crazyass}.
\begin{itemize}
\item We have that $\reciprocal= \ourceiling*{2^{N+L}/d}$ is the smallest value of $\reciprocal$ satisfying $2^{N+L} \leq \reciprocal \ourtimes d $.
\item Furthermore, when $d$ does not divide $2^{N+L}$, we have that $\ourceiling*{2^{N+L}/d} \ourtimes d = 2^{N+L} +  d-(2^{N+L} \bmod d)$
and so $\reciprocal \ourtimes d \leq 2^{N+L} +2^L$ implies $d \leq (2^{N+L} \bmod d)+2^L$. See Lemma~\ref{lemma:technicalstuff2}.
\end{itemize}
On this basis,
Algorithm~\ref{algo:bigsummary} gives  the minimal number of fractional bits $F$.
It is sufficient to pick $F\geq  N + \log_2(d)$.

\begin{lemma}\label{lemma:technicalstuff2}Given a divisor $d\in [1,2^N)$, if we set $\reciprocal=\ourceiling*{2^{N+L}/d}$, then

\begin{itemize}
\item  $\reciprocal \ourtimes d = 2^{N+L} +d -(2^{N+L} \bmod d)  $
 when $d$ is not a power of two,
\item and $\reciprocal \ourtimes d = 2^{N+L}$ when $d$ is a power of two.
\end{itemize}
\end{lemma}
\begin{proof}
The case when $d$ is a power of two follows by inspection, so suppose that $d$ is not a power of two.
 We seek to round $2^{N+L}$ up to the next multiple of $d$.    
 The previous multiple of $d$ is smaller than $2^{N+L}$ by  $2^{N+L} \bmod d$.
  Thus we need to add $d - 2^{N+L} \bmod d$ to  $2^{N+L}$ to get the next multiple of $d$.
 \end{proof}

\begin{example}
Consider Table~\trimmedref{tab:example6} where we divide by  $d=6$ and we want to support the numerators $n$ between 0 and $64 = 2^6$ so that $N=6$.
It is enough to pick $F\geq  N + \log_2(d) \approx 8.58 $ or $F=9$ but we can do better.
According to Algorithm~\ref{algo:bigsummary}, the number of fractional bits $F=N+L=6+L$ must satisfy $d \leq (2^{6+L} \bmod d)+2^L$. 
Picking $L=0$ and $F=6$ does not work
since $2^6 \bmod 6 + 1 = 5$. Picking $L=1$ also does not work since $2^7 \bmod 6 + 2 = 4$.
Thus we need $L=2$ and $F=8$, at least. So we can pick $\reciprocal = \ourceiling*{2^8 / 6} = 43$. In binary, representing 43 with 
8 fractional bits gives $0.00101011$, as in Table~\ref{tab:example6}.
Let us divide $63$ by $6$. The quotient is $(63 \ourtimes 43) \ourdiv 2^8 = 10$.
The remainder is $( ( (63 \ourtimes 43) \bmod 2^8 ) \ourtimes 6) \ourdiv 2^8 = 3$.
\end{example}

It is not always best to use the smallest number of fractional bits.
For example, we can always conveniently pick $F=2\ourtimes N$ (meaning $L=N$) and $\reciprocal= \ourceiling*{2^{2 N}/d}$,
since $d \leq 2^F \bmod d + 2^N$ clearly holds (given $d \leq 2^N$).

\begin{algorithm}
\begin{algorithmic}[1]
\State \textbf{Require}: fixed integer divisor $d \in [1,2^{N})$
\State We seek the smallest number of fractional bits $F$ such that for
any integer numerator $n \in [0,2^{N})$, the  remainder is $r = \left ( \left ( \left ( \reciprocal \ourtimes n  \right ) \bmod 2^F \right )\ourtimes d \right ) \ourdiv 2^F$ for some scaled approximate reciprocal $\reciprocal$.
\State We can always choose the scaled approximate reciprocal $\reciprocal \leftarrow \ourceiling*{2^F/d}$.
\If{$d$ is a power of two}
\State Let $F\leftarrow \log_2(d)$ and $\reciprocal=1$.
\Else
\State Let $F\leftarrow N + L$ where $L$ is the smallest integer such that $d \leq (2^{N+L} \bmod d)+2^L$.
\EndIf
\end{algorithmic}
\caption{Algorithm to select the number of fractional bits and the scaled approximate reciprocal in the case of unsigned integers.  \label{algo:bigsummary}}
\end{algorithm}

\subsection{Directly Computing Signed Remainders}
\label{sec:directlycomputingsignedremainders}

Having established that we could compute the remainder in the unsigned case directly, without first computing the quotient, we proceed to establish the same in the signed integer case.
We assume throughout that the processor represents signed integers in $[-2^{N-1}, 2^{N-1})$ using the two's complement notation. We assume that the integer $N\geq 1$.

Though the quotient and the remainder have a unique definition over positive integers, there are several valid ways to define them over signed integers. We adopt a convention regarding signed integer arithmetic that is widespread among modern computer languages, including C99, Java, C\#, Swift, Go, and Rust.
Following Granlund and Montgomery~\cite{Granlund:1994:DII:773473.178249}, we let $\trunc(v) $ be $v$ rounded towards zero: it is $\ourfloor*{v}$ if $v\geq 0$ and $\ourceiling*{v}$ otherwise. 
We use ``$\signeddiv$'' to denote the signed integer division defined as $n \signeddiv d \equiv \trunc(n/d)$
and ``$\signedmod$'' to denote the signed integer remainder defined by the identity  $n \signedmod d \equiv  n - \trunc(n/d) \ourtimes d$.
Changing the sign of the divisor changes the sign of the quotient but the remainder is insensitive to the sign of the divisor. Changing the sign of the numerator changes the sign of both the quotient and the remainder:  $(-n) \signeddiv d = -(n \signeddiv d)$   and  $(-n) \signedmod d = -(n \signedmod d)$.

Let $\lsb_K(n)$ be the function that selects the $K$~least significant bits of an integer $n$, zeroing others. The result is always non-negative (in $[0, 2^K)$) in our work.
Whenever $2^K$ divides $n$, whether $n$ is positive or not, we have  $\lsb_K(n) = 0 = n \bmod 2^K$. 

\begin{remark}\label{remark:clever}
Suppose $2^K$ does not divide $n$, then $ n \bmod 2^K =  2^K - \lsb_K(n)$ when the integer $n$ is negative, and $ n \bmod 2^K = \lsb_K(n)$ when it is positive.  Thus we have $\lsb_K(n)+\lsb_K(-n) = 2^K$ whenever  $2^K$ does not divide $n$.
\end{remark}

We establish a few technical results before proving Theorem~\ref{theorem:bettersigned}. We believe it is novel.

\begin{lemma}\label{lemma:moreclever}Given $d\in [1,2^{N-1})$, and non-negative integers $\reciprocal, L $ such that \begin{eqnarray*}2^{N-1+L} < \reciprocal \ourtimes d < 2^{N-1+L} +2^L,\end{eqnarray*}
we have that  $2^{N-1+L}$ cannot divide $\reciprocal \ourtimes n$ for any $n\in [-2^{N-1},0)$.
 \end{lemma}
\begin{proof}
First, we prove that  $2^{L}$ cannot divide $\reciprocal$.
When $2^L$ divides $\reciprocal$, setting $\reciprocal=\alpha 2^{L}$ for some integer $\alpha$, we have that  
$2^{N-1} < \alpha  d < 2^{N-1} +1$, but there is no integer between 
$2^{N-1}$ and $2^{N-1} +1$.

Next, suppose that $n\in [-2^{N-1},0)$ and $2^{N-1+L}$ divides $\reciprocal \ourtimes n$. 
Since $2^{N-1+L}$ divides $\reciprocal \ourtimes n$, we know that  the prime factorization of $\reciprocal \ourtimes n$ has at least $N-1+L$
copies of 2.   Within the range of $n$ ($[-2^{N-1},0)$) at most $N-1$ copies of 2 can be provided.
Obtaining the required $N-1+L$ copies of 2 is only possible when $n=-2^{N-1}$ and $\reciprocal$ provides
the remaining copies---so $2^L$ divides $\reciprocal$. But that is impossible.
\end{proof}

\begin{lemma}\label{lemma:nofun} Given $d\in [1,2^{N-1})$, and non-negative integers $\reciprocal, L $ such that \begin{eqnarray*}2^{N-1+L} < \reciprocal \ourtimes d < 2^{N-1+L} +2^L,\end{eqnarray*} 
for all integers $n\in (0,2^{N-1}]$ and letting
$y = \lsb_{N-1+L}(\reciprocal \ourtimes n)$, 
we have that  $  (y \ourtimes d) \bmod  2^{N-1+L}  >0$.

 \end{lemma}
\begin{proof}
Write $n= q \ourtimes d + r$ where $r\in [0,d)$.

Since $\reciprocal \ourtimes n$ is positive 
%
we have that $y = (\reciprocal \ourtimes n) \bmod 2^{N-1+L} = \lsb_{N-1+L}(\reciprocal \ourtimes n)$. (See Remark~\ref{remark:clever}.)

Lemma~\ref{lemma:interval} is applicable, replacing $N$ with $N-1$. We have a stronger constraint on $\reciprocal \ourtimes d$, but that is not harmful. 
Thus we have $y \in [ \reciprocal \ourtimes r, \reciprocal \ourtimes r + 2^L q]$.

We proceed as in the proof of  Theorem~\ref{theorem:crazyass}.
We want to show that $d\ourtimes y \in (2^{N-1+L} r, 2^{N-1+L} (r+1))$. 
\begin{itemize}
\item ($d\ourtimes y > 2^{N-1+L} r$) We have $	y\geq \reciprocal \ourtimes r$. Multiplying throughout by $d$, we get
$d\ourtimes y  \geq \reciprocal \ourtimes d \ourtimes r> 2^{N-1+L} r$, where we used $ \reciprocal \ourtimes d > 2^{N-1+L}$ in the last inequality.
\item ($d\ourtimes y  < 2^{N-1+L} (r+1)$) We have $	y\leq  \reciprocal \ourtimes r + 2^L q$.
 Multiplying throughout by $d$, we get
$d\ourtimes y \leq  \reciprocal \ourtimes  d \ourtimes r+ 2^L q \ourtimes d < 2^{N-1+L} r + 2^L n$. Because $n \leq 2^{N-1}$, we have the result $d\ourtimes y  < 2^{N-1+L} (r+1)$.
\end{itemize}

Thus we have 
$d \ourtimes y \in (2^{N-1+L} r, 2^{N-1+L} (r+1))$, which shows that  $  (y \ourtimes d) \bmod  2^{N-1+L}  >0$.
 This completes the proof.
\end{proof}

\begin{lemma}\label{lemma:rarara} 

Given positive integers $a, b, d$, we have that
$\ourfloor*{( b - a) \ourtimes d / b }
= d - 1 - 
\ourfloor*{a \ourtimes d / b }$
if $b$ does not divide $a \ourtimes d $.
 \end{lemma}
 \begin{proof}
 Define $p = a \ourtimes d / b - \ourfloor*{a \ourtimes d / b }$. We
 have $p\in (0,1)$.
 We have that 
 $\ourfloor*{( b - a) \ourtimes d / b }
 = \ourfloor*{d -  a \ourtimes d / b }
 =  \ourfloor*{d -  \ourfloor*{a \ourtimes d / b }  - p }
 = d - 1 - \ourfloor*{a \ourtimes d / b }.$
 \end{proof}

\begin{theorem}\label{theorem:bettersigned} Given $d\in [1,2^{N-1})$, and non-negative integers $\reciprocal, L $ such that \begin{eqnarray*}2^{N-1+L} < \reciprocal \ourtimes d < 2^{N-1+L} +2^L,\end{eqnarray*} 
let $\mu= \ourfloor*{( \lsb_{N-1+L}(\reciprocal \ourtimes n)  \ourtimes d ) / 2^{N-1+L} }$ then
\begin{itemize}
\item $n \signedmod d = \mu $
for all $n \in [0,2^{N-1})$
\item and $n \signedmod d =  \mu -d + 1$
for all $n \in [-2^{N-1},0)$.
\end{itemize}
\end{theorem}
\begin{proof}
When $n$ is non-negative, then so is $\reciprocal \ourtimes n$, and 
$\lsb_{N-1+L}(\reciprocal \ourtimes n) $ is equal to $(\reciprocal \ourtimes n) \bmod 2^{N-1+L}$; 
there is no distinction between signed and unsigned $\bmod$, 
so the result follows by Theorem~\ref{theorem:crazyass}, replacing $N$ by $N-1$. (Theorem~\ref{theorem:crazyass} has a weaker constraint on $\reciprocal \ourtimes d$.)

Suppose that $n$ is negative ($n \in [-2^{N-1},0)$).
By Lemma~\ref{lemma:moreclever}, $2^{N-1+L}$ cannot divide 
$\reciprocal \ourtimes n$.
Hence, we have that 
$\lsb_{N-1+L}(\reciprocal \ourtimes n)   = 2^{N-1+L} - \lsb_{N-1+L}(\reciprocal \ourtimes (-n))$ by Remark~\ref{remark:clever}.
Thus we have
\begin{eqnarray*}\mu & =  &\ourfloor*{\lsb_{N-1+L}(\reciprocal \ourtimes n)  \ourtimes d / 2^{N-1+L} }\\
 &  = & \ourfloor*{\left( 2^{N-1+L} - \lsb_{N-1+L}(\reciprocal \ourtimes (-n))  \right) \ourtimes d / 2^{N-1+L}  } \text{\hspace{1cm}  by Remark~\ref{remark:clever}}\\
 & = & d - 1  - \ourfloor*{\lsb_{N-1+L}(\reciprocal \ourtimes (-n)) \ourtimes d  / 2^{N-1+L}  }\text{\hspace{1cm}by Lemmata~\ref{lemma:nofun}~and~\ref{lemma:rarara}}\\
  & = & d - 1 - \ourfloor*{\left((\reciprocal \ourtimes (-n)) \bmod 2^{N-1+L}\right) \ourtimes d / 2^{N-1+L}  }\\
  & = & d - 1 -  ((-n) \bmod d) \text{\hspace{5.3cm}  by Theorem~\ref{theorem:crazyass}.}\\
 \end{eqnarray*}
 Hence we have $\mu -d + 1 = -  ((-n) \bmod d) = n \bmod d$, which concludes the proof.
 \end{proof}

We do not need to be concerned with negative divisors since $n \signedmod d = n \signedmod -d$ for all integers $n$.

We can pick $\reciprocal, L$  in a manner similar to the unsigned case. 
We can  choose $\reciprocal= \ourfloor*{2^F/d} + 1$ and let $F= N - 1 + L$ where $L$ is an integer such that  $2^{N-1+L} < \reciprocal \ourtimes d < 2^{N-1+L} +2^L$.
With this choice of $\reciprocal$, we 
have that $\reciprocal \ourtimes d = 2^{N-1+L} - (2^{N-1+L} \bmod d) + d$. Thus we have 
the constraint $-(2^{N-1+L} \bmod d) + d<2^L$ on $L$. Because $- (2^{N-1+L} \bmod d) + d \in [1,d]$, it  suffices to pick $L$ large enough so that $2^L>d$. Thus any  $L> \log_2 d $ would do, and hence $F>  N + \log_2(d)$ is sufficient.
 It is not always best to pick $L$ to be minimal: it could be convenient to pick $L=N+1$.

\subsection{Fast Divisibility Check with a Single Multiplication}
\label{sec:fastdivcheck}

Following earlier work by Artzy et al.\cite{Artzy:1976:FDT:359997.360013},
Granlund and Montgomery\cite{Granlund:1994:DII:773473.178249} describe how we can check quickly whether an unsigned integer is divisible by a constant, without computing the remainder. We summarize their approach
before providing an alternative.
Given an odd divisor $d$, we can find its (unique) multiplicative inverse $\bar d$ defined as $d \ourtimes \bar d \bmod 2^N = 1$.
The existence of a multiplicative inverse $\bar d$ allows us to quickly divide an integer $n$ by $d$ when it is divisible by $d$, if $d$ is odd.
It suffices to multiply $n=a\ourtimes d$ by $\bar d$:
$n \ourtimes \bar d \bmod 2^N = a \ourtimes (d \ourtimes \bar d) \bmod 2^N = a \bmod 2^N = n \ourdiv d$.
When  the divisor is $2^K d$ for $d$ odd and $n$ is divisible by $2^K d$, then we can write $ n \ourdiv (2^K d) = (n \ourdiv 2^K) \ourtimes \bar d \bmod 2^N $.
As pointed out by Granlund and Montgomery, this observation can also enable us to quickly check whether a number is divisible by $d$. If $d$ is odd and $n\in [0,2^N)$ is divisible by $d$, then $n \ourtimes \bar d \bmod 2^N  \in [0, \ourfloor*{(2^N-1)/d }]$. Otherwise $n$ is not divisible by $d$. Thus, when $d$ is odd, we can check whether any integer in $[0,2^N)$ is divisible by $d$ with a multiplication followed by a comparison.
When the divisor is even ($2^K d$), then we have to check that $n \ourtimes \bar d \bmod 2^N  \in  [0, 2^K \ourtimes \ourfloor*{(2^N-1)/d }]$
and that $n$ is divisible by $2^K$ (i.e., $n \bmod 2^K = 0$). We can achieve the desired result by computing $n \ourtimes \bar d$, rotating the resulting word by $K$~bits and comparing the result with $\ourfloor*{(2^N-1)/d }$.

Granlund and Montgomery can check that an unsigned integer is divisible by another using as little as one multiplication and comparison when the divisor is odd, and a few more instructions when the divisor is even. Yet we can always check the divisibility with a single multiplication and a modulo reduction to a power of two---even when the divisor is even because of the following proposition. Moreover, a single precomputed constant ($\reciprocal$) is required.

\begin{proposition}\label{prop:divisibility}Given $d\in [1,2^N)$, and non-negative integers $\reciprocal, L $ such that 
$2^{N+L} \leq \reciprocal \ourtimes d \leq 2^{N+L} +2^L$ then given some $n \in [0,2^N)$, we have that
$d$ divides $n$ if and only if $(\reciprocal \ourtimes n) \bmod 2^{N+L} < \reciprocal$.
\end{proposition}
\begin{proof}
We have that $d$ divides $n$ if and only if $n \bmod d = 0$.
By Lemma~\ref{lemma:interval}, we have that
$\reciprocal \ourtimes (n\bmod d) \leq (\reciprocal \ourtimes n) \bmod 2^{N+L} \leq  \reciprocal \ourtimes (n\bmod d) + 2^{L} (n \ourdiv d)$.
We want to show that  $n \bmod d = 0$ is equivalent to $(\reciprocal \ourtimes n) \bmod 2^{N+L} < \reciprocal$.

Suppose that $n \bmod d = 0$, then we have that $(\reciprocal \ourtimes n) \bmod 2^{N+L} \leq  2^{L} (n \ourdiv d)$. However, by our constraints on $\reciprocal$, we have that $\reciprocal \geq 2^{N+L}/d>2^{L} (n \ourdiv d)$. Thus, if $n \bmod d = 0$ then $(\reciprocal \ourtimes n) \bmod 2^{N+L} < \reciprocal$.

Suppose that $(\reciprocal \ourtimes n) \bmod 2^{N+L} <\reciprocal$, then because $\reciprocal \ourtimes (n\bmod d) \leq (\reciprocal \ourtimes n) \bmod 2^{N+L} $, we have that
$\reciprocal \ourtimes (n\bmod d)< \reciprocal$ which implies that $n \bmod d = 0$.
This completes the proof.
\end{proof}

Thus if we have a reciprocal $\reciprocal= \ourceiling*{2^{F}/d}$  with $F=N+L$  large enough to compute the remainder exactly (see Algorithm~\ref{algo:bigsummary}), then $(\reciprocal \ourtimes n) \bmod 2^{F} < \reciprocal$ if and only if $n$ is divisible by $d$.
We do not need to pick $F$ as small as possible.
In particular, if we set $\reciprocal= \ourceiling*{2^{2N}/d}$, 
then $(\reciprocal \ourtimes n) \bmod 2^{2N} < \reciprocal$ if and only if $n$ is divisible by $d$.

\begin{remark}\label{remark:signeddivisibility}
We can extend our fast divisibility check to the signed case. Indeed, we have that $d$ divides $n$ if and only if $|d|$ divides $|n|$. Moreover, the absolute
value of any $N$-bit negative integer can be represented as an $N$-bit unsigned integer.
\end{remark}

\section{Software Implementation}
\label{sec:softimpl}
Using the C language,
we provide our implementations of the 32-bit remainder computation (i.e., \texttt{a \% d}) in Figs.~\trimmedref{fig:codemod}~and~\trimmedref{fig:codesignedmod} for unsigned and signed integers. In both case, the programmer is expected to precompute the constant \texttt{c}. For simplicity, the code shown here explicitly does not  handle the divisors $d\in\{-1,0,1,-2^{31}\}$.


For the x64 platforms, we provide the instruction sequences in assembly code produced by GCC and Clang for computing $n \bmod 95$ in Fig.~\trimmedref{fig:codesample}; in the third column, we provide the x64 code produced with our approach after constant folding. Our approach generates about half as many instructions.

In Fig.~\trimmedref{fig:codesample-arm}, we make the same comparison on the 64-bit ARM platform, putting side-by-side  compiler-generated code for the Granlund-Montgomery-Warren approach with   code generated from our approach.  As a RISC processor, ARM does not
handle most large constants in a single machine instruction, but typically
assembles them from 16-bit quantities. Since the Granlund-Montgomery-Warren algorithm requires only 32-bit constants, two 16-bit values are sufficient whereas our approach relies on 64-bit quantities and thus needs four  16-bit values. The ARM processor also has a ``multiply-subtract'' instruction that
is particularly convenient for computing the remainder from the quotient. Unlike the case with x64, our approach does not translate into significantly fewer instructions on the ARM platform.

These code fragments show that a code-size saving is achieved by our
approach on x64 processors, compared to the approach taken by the compilers.
We verify in \S~\ref{sec:experiments} that there is also a runtime advantage.

\begin{figure}\centering
\begin{minipage}{0.8\textwidth}
\begin{lstlisting}
uint32_t d = ...;// your divisor > 0
// c = ceil( (1<<64) / d ) ; we take L = N
uint64_t c = UINT64_C(0xFFFFFFFFFFFFFFFF) / d + 1;
// fastmod computes (n mod d) given precomputed c
uint32_t fastmod(uint32_t n /* , uint64_t c, uint32_t d */) {
  uint64_t lowbits = c * n;
  return ((__uint128_t)lowbits * d) >> 64;
}\end{lstlisting}
\end{minipage}
\caption{\label{fig:codemod}C code implementing a fast unsigned remainder function using the \texttt{__uint128_t} type extension.
}
\end{figure}

\begin{figure}\centering
\begin{minipage}{0.8\textwidth}
\begin{lstlisting}
int32_t d = ...;// your non-zero divisor in [-2147483647,2147483647]
uint32_t pd =  d < 0 ? -d : d; // absolute value, abs(d)
// c = floor( (1<<64) / pd ) + 1; Take L = N + 1
uint64_t c = UINT64_C(0xFFFFFFFFFFFFFFFF) / pd
                  + 1 + ((pd & (pd-1))==0 ? 1 : 0);
// fastmod computes (n mod d) given precomputed c
int32_t fastmod(int32_t n /* , uint64_t c, uint32_t pd */) {
  uint64_t lowbits = c * n;
  int32_t highbits = ((__uint128_t) lowbits * pd) >> 64;
  // answer is equivalent to (n<0) ? highbits - 1 + d : highbits
  return highbits - ((pd - 1) & (n >> 31));
}
\end{lstlisting}
\end{minipage}
\caption{\label{fig:codesignedmod}C code implementing a fast signed remainder function using the \texttt{__uint128_t} type extension.
}
\end{figure}

\begin{figure}\centering
{
\begin{tabular}{c|c|c}
\begin{lstlisting}
// GCC 6.2
mov	eax, edi
mov	edx, 1491936009
mul	edx
mov	eax, edi
sub	eax, edx
shr	eax
add	eax, edx
shr	eax, 6
imul	eax, eax, 95
sub	edi, eax
mov	eax

\end{lstlisting}
&
\begin{lstlisting}
// Clang 4.0
mov	eax, edi
imul	rax, rax, 1491936009
shr	rax, 32
mov	ecx, edi
sub	ecx, eax
shr	ecx
add	ecx, eax
shr	ecx, 6
imul	eax, ecx, 95
sub	edi, eax
mov	eax, edi
\end{lstlisting}
&
\begin{lstlisting}
// our fast version  
// + GCC 6.2
movabs	rax, 
  194176253407468965
mov	edi, edi
imul	rdi, rax
mov	eax, 95
mul	rdi
mov	rax, rdx
//
//
//
//
\end{lstlisting}
\end{tabular}
}
\caption{\label{fig:codesample}Comparison between the x64 code generated by GCC 6.2 for unsigned $a \bmod 95$ (left) and our version (right).
Clang 4.0 generated the middle code, and when compiling our version (not shown)  used a \texttt{mulx} instruction to place the high bits of the product directly into the
return register, saving one instruction over GCC\@.}
\end{figure}

\begin{figure}\centering
\begin{tabular}{c|c}
\begin{lstlisting}
// GCC 6.2 for a % 95 on ARM
mov	w1, 8969
mov	w3, 95
movk	w1, 0x58ed, lsl 16
umull	x1, w0, w1
lsr	x1, x1, 32
sub	w2, w0, w1
add	w1, w1, w2, lsr 1
lsr	w1, w1, 6
msub	w0, w1, w3, w0
\end{lstlisting}
&
\begin{lstlisting}
// our version of a % 95 + GCC 6.2
mov	x2, 7589
uxtw	x0, w0
movk	x2, 0x102b, lsl 16
mov	x1, 95
movk	x2, 0xda46, lsl 32
movk	x2, 0x2b1, lsl 48
mul	x0, x0, x2
umulh	x0, x0, x1
//

\end{lstlisting}
\end{tabular}
\caption{\label{fig:codesample-arm}Comparison between the ARM code generated by GNU GCC 6.2 for $a \bmod 95$ (left) and our word-aligned version (right).
In both cases, except for instruction order, Clang's code was similar to GCC's.}

\end{figure}

\subsection{Divisibility}
\label{sec:divisibility-impl}

We are interested in determining quickly whether a 32-bit integer $d$ divides a 32-bit integer $n$---faster than by checking whether the remainder is zero.
To the best of our knowledge, no compiler includes such an optimization, though some software libraries provide related  fast functions.\footnote{\url{https://gmplib.com}}
We present the code for our approach (LKK) in Fig.~\trimmedref{fig:lkk-unsigned-impl}, and our implementation of the Granlund-Montgomery approach (GM) in Fig.~\trimmedref{fig:gm-unsigned-impl}.

\begin{figure}\centering
\begin{minipage}{0.65\textwidth}
\begin{lstlisting}
// calculate c for use in lkk_divisible
uint64_t lkk_cvalue(uint32_t d) {
  return 1 + UINT64_C(0xffffffffffffffff) / d;
}

// given precomputed c, checks whether n % d == 0
bool lkk_divisible(uint32_t n,
    uint64_t c) {
  // rhs is large when c==0
  return n * c <= c - 1;
}
\end{lstlisting}
\end{minipage}
\caption{\label{fig:lkk-unsigned-impl} Unsigned divisibility test, our approach.}
\end{figure}

\begin{figure}\centering
\begin{minipage}{0.6\textwidth}
  \begin{lstlisting}
// rotate n by e bits, avoiding undefined behaviors
// cf https://blog.regehr.org/archives/1063
uint32_t rotr32(uint32_t n, uint32_t e) {
  return (n >> e) | ( n << ( (-e)&31) );
}

// does d divide n?
// d = 2**e * d_odd; dbar = multiplicative_inverse(d_odd)
// thresh = 0xffffffff / d
bool gm_divisible(uint32_t n,
    uint32_t e, uint32_t dbar,
    uint32_t thresh) {
  return rotr32(n * dbar, e) <= thresh;
}

// Newton's method per Warren,
// Hacker's Delight  pp. 246--247
uint32_t multiplicative_inverse(uint32_t d) {
  uint32_t x0 = d + 2 * ((d+1) & 4);
  uint32_t x1 = x0 * (2 - d * x0);
  uint32_t x2 = x1 * (2 - d * x1);
  return  x2 * (2 - d * x2);
}
\end{lstlisting}
\end{minipage}
\caption{\label{fig:gm-unsigned-impl} Unsigned divisibility test, Granlund-Montgomery approach.}
\end{figure}

\section{Experiments}

\label{sec:experiments}


Superiority over the Granlund-Montgomery-Warren approach might depend on such CPU characteristics
as the relative speeds of instructions for integer division, 32-bit integer multiplication and 64-bit integer division.
Therefore, we tested our software on several x64 platforms and on ARM\footnote{
With GCC~4.8 on the ARM platform we observed that, for many constant divisors, the
compiler chose to generate a \texttt{udiv} instruction instead of using the
Granlund-Montgomery code sequence.  This is not seen for GCC~6.2.
} and POWER8 servers, and relevant details are given in
Table~\trimmedref{tab:test-cpus}.
The choice of
multiplication instructions and instruction scheduling can vary by compiler, and thus we tested using
various versions of GNU GCC and LLVM's Clang.
For brevity we primarily report results from the Skylake platform, with
comments on points where the other platforms were significantly different.
For the  Granlund-Montgomery-Warren approach with compile-time constants, we use the optimized divide and remainder operations built into
GCC and Clang. 

We sometimes need to repeatedly divide by a constant that is known only at runtime. In such instances, an optimizing compiler
may not be helpful.   Instead a programmer might rely on a library offering fast division functions. 
For runtime constants on x64 processors, we use the libdivide library\footnote{\url{http://libdivide.com}} as it provides a 
well-tested and optimized implementation.

On x64 platforms, we use the compiler flags  \texttt{-O3 -march=native}; on ARM
we use \texttt{-O3 -march=armv8-a} and on
POWER8 we use \texttt{-O3 -mcpu=power8}.
Some tests have results reported in wall-clock time, whereas
in other tests, the Linux \texttt{perf stat} command was used to obtain the total number of processor cycles
spent doing an entire benchmark program.
To ease reproducibility, we make our benchmarking software and scripts freely available.\footnote{\url{https://github.com/lemire/constantdivisionbenchmarks}}

\begin{table}
\caption{\label{tab:test-cpus} Systems Tested
}
\centering
\begin{minipage}{\textwidth}
\centering
\begin{tabular}{cccc}\toprule
Processor      & Microarchitecture                             & Compilers\\ \midrule
Intel i7-6700  & Skylake (x64) 
                                                               & GCC 6.2; Clang 4.0   & default platform    \\
Intel i7-4770  & Haswell (x64)                                       & GCC  5.4; Clang 3.8  &      \\
AMD Ryzen 7 1700X  & Zen (x64)                                        & GCC 7.2; Clang 4.0  &  \\
POWER8 & POWER8 Murano & GCC 5.4; Clang 3.8\\
AMD Opteron A1100 & ARM Cortex A57 (Aarch64) &  GCC 6.2; Clang 4.0   &  \\
\bottomrule
\end{tabular}
\end{minipage}
\end{table}


\subsection{Beating the Compiler}
\label{sec:beating-compiler}

We implement a 32-bit 
linear congruential generator\cite{knuth2} that generates random numbers according to the function $X_{n+1} = (a \ourtimes X_n +b) \bmod d$, starting from a given seed $X_0$. 
Somewhat arbitrarily, we set the seed to 1234, we use 31 as the multiplier ($a=31$) and the additive constant is set to 27961 ($b=27961$). We call the function 100~million times, thus generating 100~million random numbers. The divisor $d$ is set at compile time.
See Fig.~\ref{fig:lcg-impl}.
In the signed case, we use a negative multiplier ($a=-31$).

\begin{figure}\centering
\begin{minipage}{0.5\textwidth}
\begin{lstlisting}
uint32_t x = 1234;
for(size_t i = 0; i < 100000000; i++) {
  // d may be set at compile time
  x = (32 * x + 27961) % d; 
}
\end{lstlisting}
\end{minipage}
\caption{\label{fig:lcg-impl} Code for a linear congruential generator used to benchmark division by a constant.}
\end{figure}

Because the divisor is a constant, compilers can optimize the integer division using the  Granlund-Montgomery approach. We refer to this scenario as the \emph{compiler} case. To prevent the compiler from proceeding with such an optimization and force it to repeatedly use the division instruction, we can declare the variable holding the modulus to be volatile (as per the C standard). We refer to this scenario as the \emph{division instruction} case. In such cases, the compiler is not allowed to assume that the modulus is constant---even though it is.
We verified that the assembly generated by the compiler includes the division instruction and does not include expensive operations such as memory barriers or cache flushes. We verified that our wall-clock times are highly repeatable\footnote{For instance, we repeated tests 20 times for 9 divisors in Figs.~\ref{fig:beatingthecompiler}abcd
and~\ref{fig:beatingthecompilernotskylake}ab, and we observed maximum differences among the 20 trials of 4.8\,\%, 0.3\,\%, 0.7\,\%, 0.0\,\%, 0.8\,\% and 0.9\,\%,
respectively.  
}.

We present our results in  Fig.~\trimmedref{fig:beatingthecompiler} where we compare with our  alternative. In all cases, our approach is superior to the code produced by the compiler, except for powers of two in the case of GCC\@. The benefit of our functions can reach 30\%.

The performance of the compiler (labelled as \emph{compiler}) depends on the divisor for both GCC and Clang, though Clang has greater variance.
The performance of our  approach is insensitive to the divisor, except when the divisor is a power of two.

We observe that  in the unsigned case, Clang optimizes very effectively when the divisor is  a small power of two. This remains true even when we disable loop unrolling (using the \texttt{-fno-unroll-loops} compiler flag).
By inspecting the produced code, we find that Clang (but not GCC) optimizes away the multiplication entirely in the sense that, for example,  $X_{n+1} = (31 \ourtimes X_n +27961) \bmod 16$ is transformed into  $X_{n+1} = \lsb_{4}(9 - X_n)$. We find it interesting that these optimizations are applied both in the \emph{compiler} functions as well as in our functions.
Continuing with the unsigned case, we find that Clang often produces slightly more efficient compiled code than GCC for our functions, even when the divisor is not a power of two: compare Fig.~\ref{fig:beatingthecompilergcc} with Fig.~\ref{fig:beatingthecompilerclang}. However, these small differences disappear if we disable loop unrolling.

Yet, GCC seems preferable in the signed benchmark: in Figs.~\ref{fig:signedbeatingthecompilergcc} and~\ref{fig:signedbeatingthecompilerclang}, Clang is slightly less efficient than GCC, sometimes requiring \SI{0.5}{\second} to complete the computation whereas GCC never noticeably exceeds \SI{0.4}{\second}.

\begin{figure}\centering
\subfloat[GNU GCC (unsigned) \label{fig:beatingthecompilergcc}]{
\includegraphics[width=0.49\columnwidth]{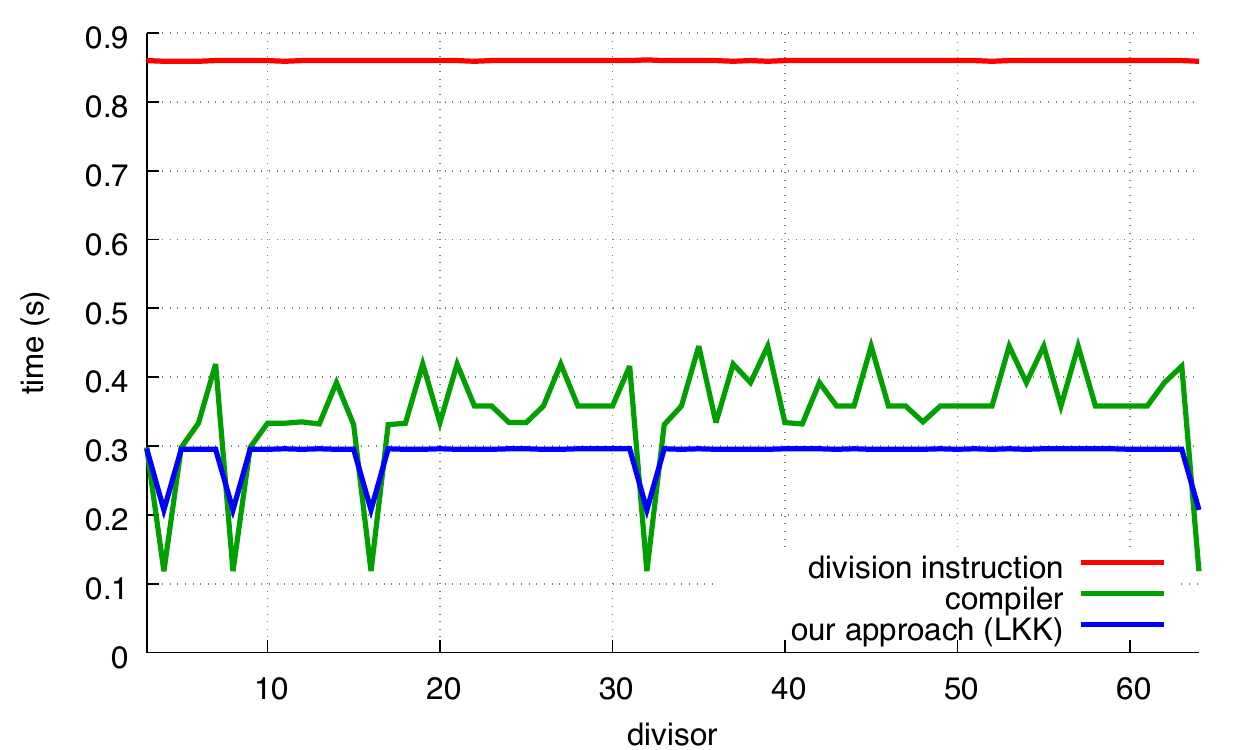}
}\subfloat[GNU GCC  (signed) \label{fig:signedbeatingthecompilergcc}]{
\includegraphics[width=0.49\columnwidth]{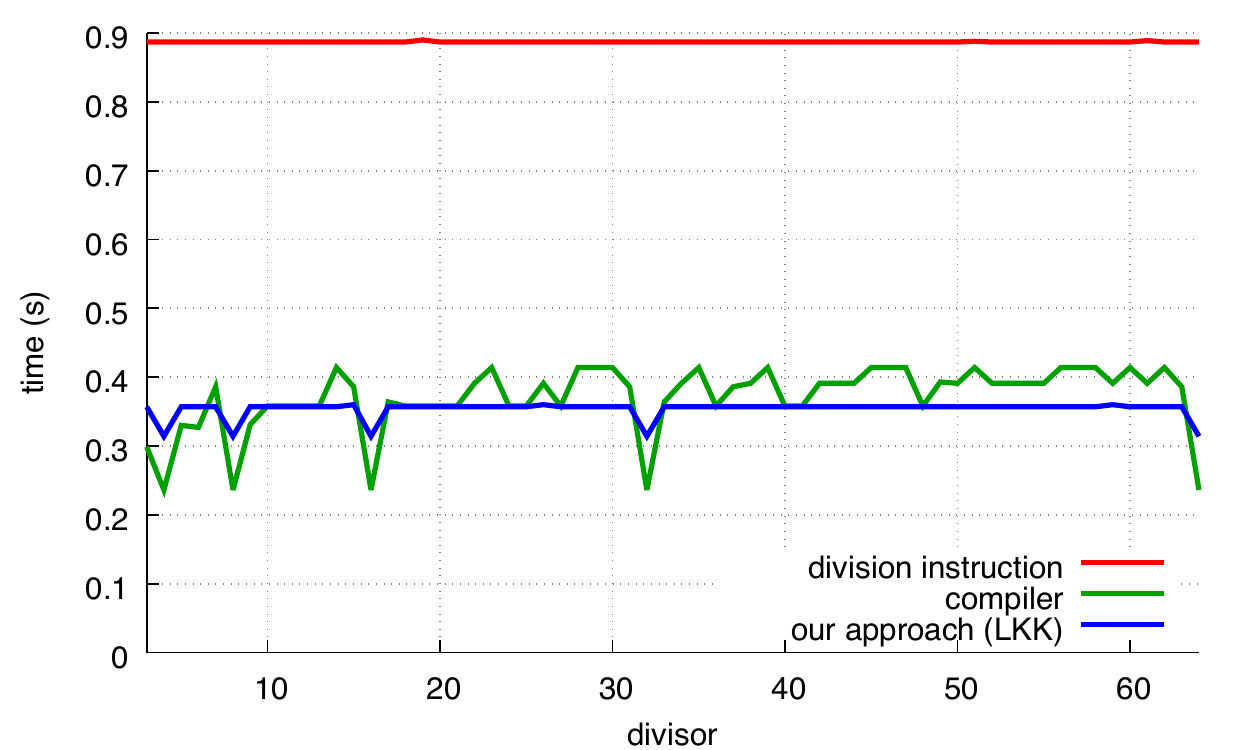}
}\\
\subfloat[LLVM's Clang (unsigned) \label{fig:beatingthecompilerclang}]{
\includegraphics[width=0.49\columnwidth]{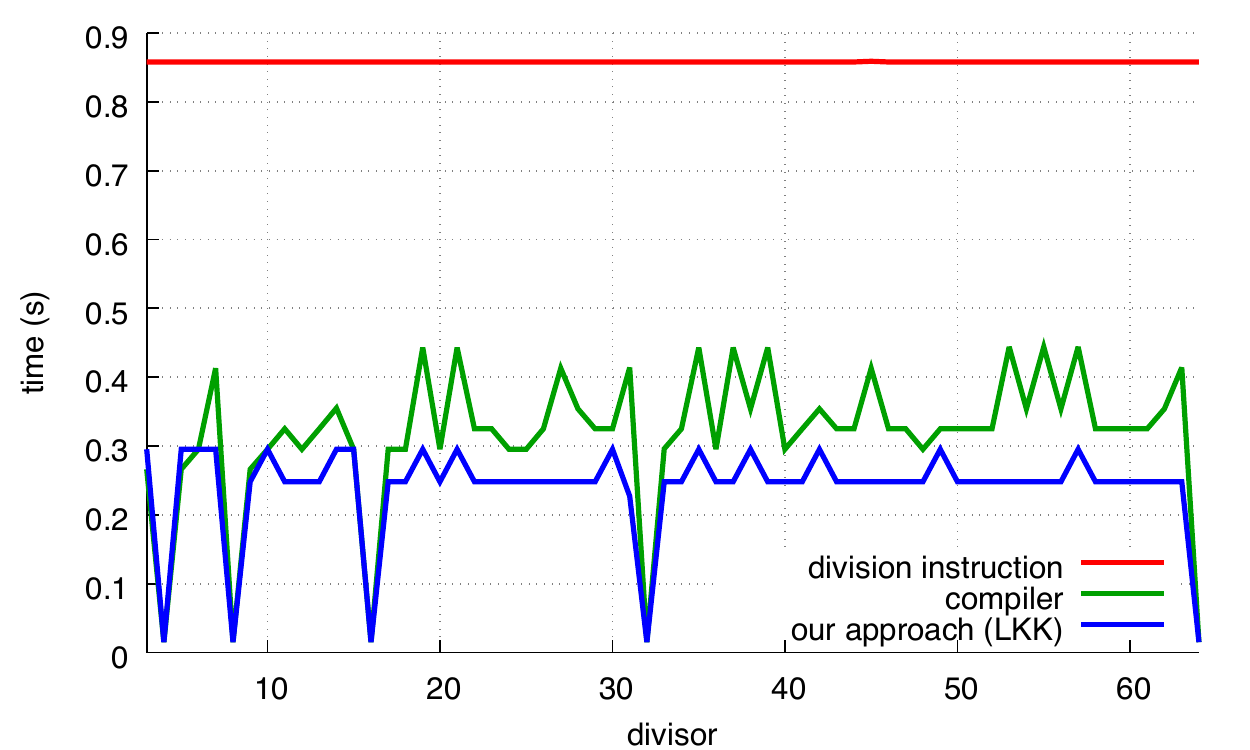}
}\subfloat[LLVM's Clang  (signed) \label{fig:signedbeatingthecompilerclang}]{
\includegraphics[width=0.49\columnwidth]{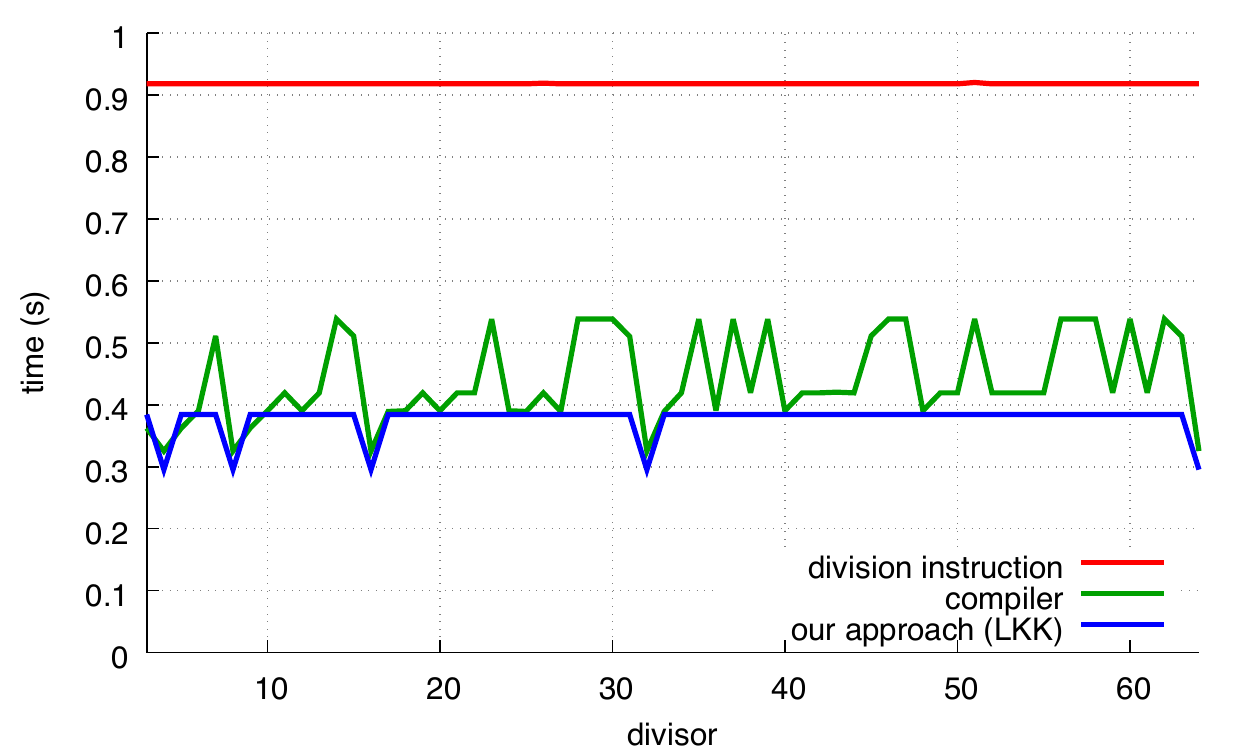}
}
\caption{\label{fig:beatingthecompiler}Wall-clock time to compute 100~million random integers using a linear congruential generator with various divisors set at compile time (Skylake x64)}
\end{figure}

\begin{figure}\centering
\subfloat[Ryzen (GCC, unsigned) \label{fig:beatingthecompilergccryzen}]{
\includegraphics[width=0.49\columnwidth]{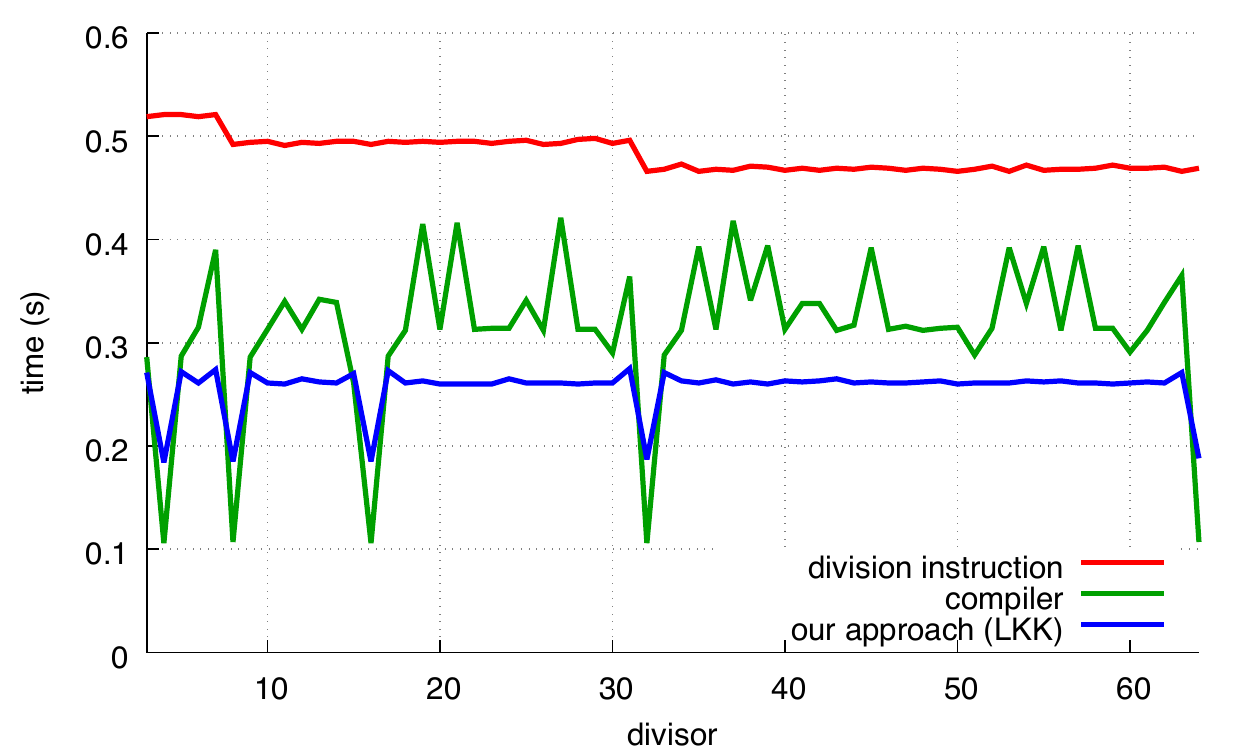}
}\subfloat[ARM (GCC, unsigned)\label{fig:beatingthecompilergccarm}]{
\includegraphics[width=0.49\columnwidth]{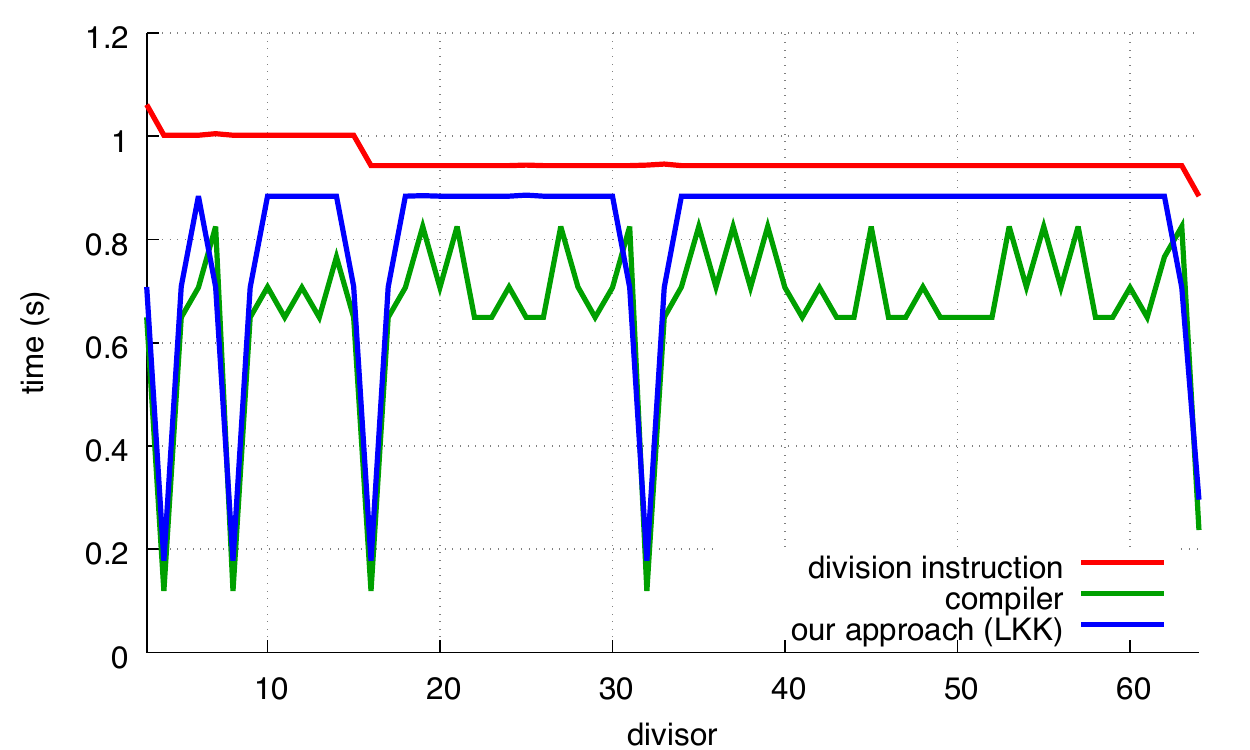}
}\\\subfloat[POWER8 (Clang, unsigned) \label{fig:beatingthecompilerclangpower}]{
\includegraphics[width=0.49\columnwidth]{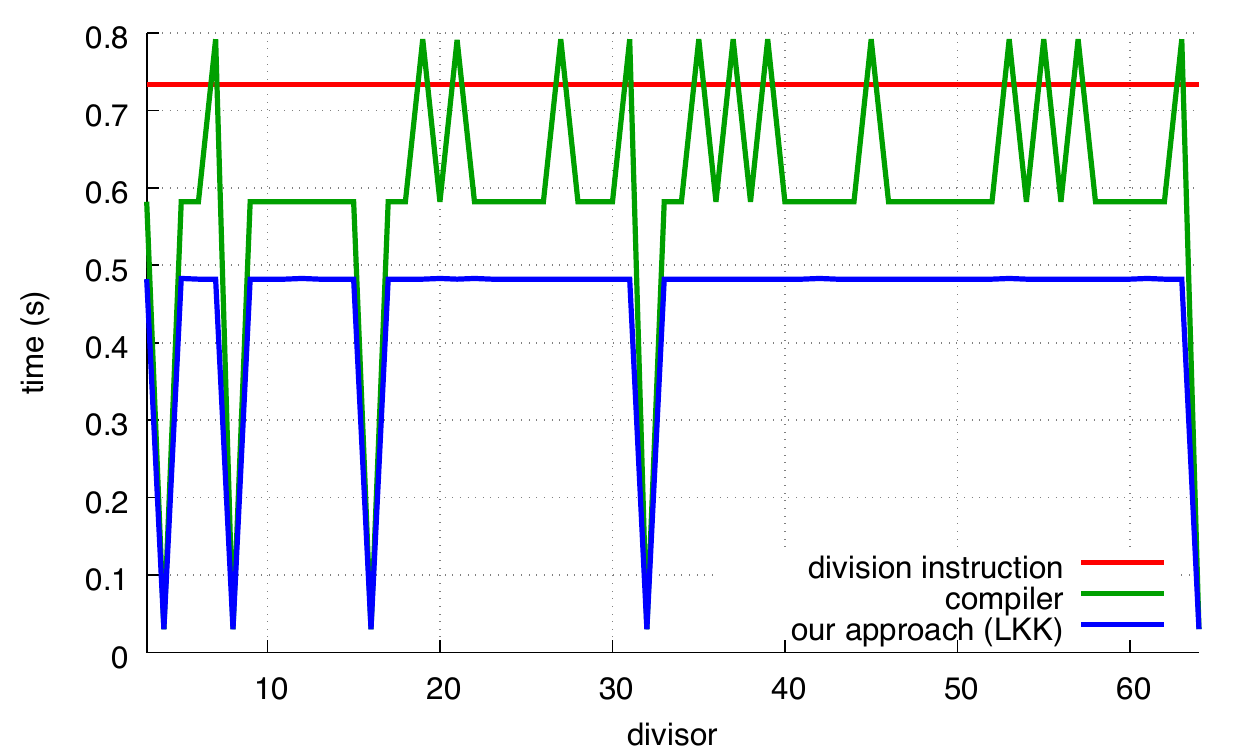}}
\subfloat[POWER8 (Clang, signed) \label{fig:beatingthecompilerclangpowersigned}]{
\includegraphics[width=0.49\columnwidth]{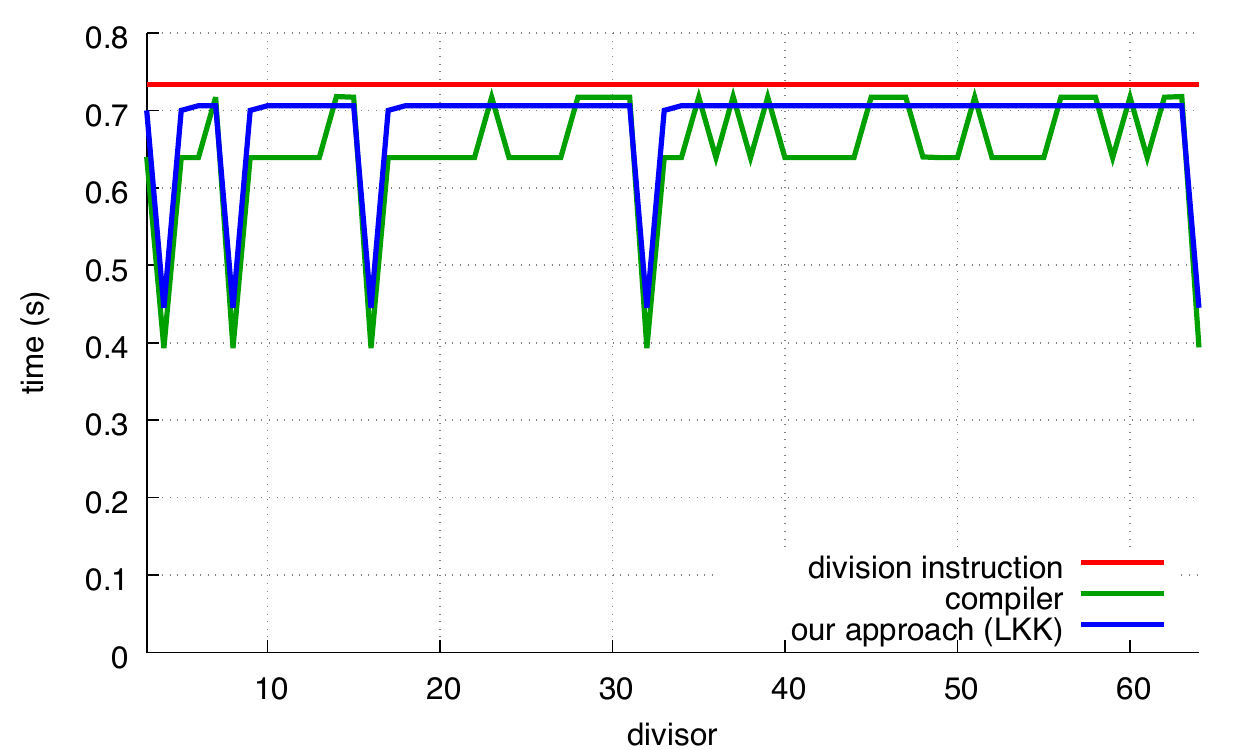}}
\caption{\label{fig:beatingthecompilernotskylake} Ryzen, ARM and POWER8 results for small divisors.}
\end{figure}

For comparison, Fig.~\trimmedref{fig:beatingthecompilernotskylake} shows how the Ryzen, POWER8 and ARM processors perform on unsigned
computations.
The speed of the hardware integer-division instruction varies,  speeding up
at $d=8$ and again at $d=32$ for Ryzen and $d=4$, 16, 64, 256 and 1024 for ARM\@.
The gap between hardware integer division  and Granlund-Montgomery (compiler) is less
on Ryzen, POWER8 and ARM than on Skylake; for some divisors, there is little benefit to using
\textit{compiler} on POWER8 and ARM\@.  On x64 platforms, our  approach continues to be significantly
faster than hardware integer division for all divisors.

On  ARM, the performance is limited when computing remainders using our approach.  
Unlike x64 processors, ARM processors require more than one instruction to load a constant such as the reciprocal ($\reciprocal$), but that is not a concern in this instance since  the compiler loads $\reciprocal$ into a register outside of the loop.
We believe that the reduced speed has to do with the
performance  of the multiplication instructions of our Cortex~A57 processor~\cite{arma57}. To compute the most significant 64~bits of a 64-bit product as needed by our functions, we must use the multiply-high instructions (\texttt{umulh} and \texttt{smulh}), but they  require six~cycles of latency and they prevent the execution of other multi-cycle instructions for an additional three~cycles. 
In contrast, multiplication instructions on x64 Skylake processors produce the full 128-bit product in three cycles.
Furthermore, our ARM processor has a multiply-and-subtract instruction with a latency of three cycles. Thus  it is advantageous to rely on the multiply-and-subtract instruction instead of the multiply-high instruction. Hence, it is faster to compute the remainder from the quotient by multiplying and subtracting ($r = n - (n \ourdiv d) \ourtimes d $). Furthermore, our ARM processor has fast division instructions:  the ARM optimization manual for Cortex A57 processors indicates that both signed and unsigned  division require between 4 and 20~cycles of latency~\cite{arma57} whereas integer division instructions on Skylake processors (\texttt{idiv} and \texttt{div}) have  26~cycles of latency for 32-bit registers~\cite{fog2016instruction}. Even if we take into account that division instructions on ARM computes solely the quotient, as opposed to both the quotient and remainder on x64, it seems that the ARM platform has a competitive division latency. Empirically, the division instruction on ARM is often within 10\% of the  Granlund-Montgomery  compiler optimization (Fig.~\ref{fig:beatingthecompilergccarm})  whereas the compiler optimization is consistently more than twice as fast as the division instruction on a Skylake processor (see Fig.~\ref{fig:beatingthecompilergcc}).  


%

Results for POWER8 are shown in Figs.~\ref{fig:beatingthecompilerclangpower}~and~\ref{fig:beatingthecompilerclangpowersigned}. Our unsigned approach is better than the compiler's; indeed the compiler would sometimes have done better to generate a divide instruction than use the Granlund-Montgomery approach.  For our signed approach, both GCC and Clang had trouble generating efficient code for many divisors. 

As with ARM, code generated for POWER8 also deals with 64-bit constants less directly than x64 processors. If not in registers, POWER8 code  loads 64-bit constants from memory, using two operations to construct a 32-bit address that is then used with a load instruction.  In this benchmark, however, the compiler keeps 64-bit constants in registers.
Like ARM, POWER8 has instructions that compute the upper
64~bits of a 64-bit product.  The POWER8 microarchitecture\cite{power8uarch} has good support
for integer division: it has two fixed-point
pipelines, each containing a multiplier unit and a divider unit. 
When the multi-cycle divider unit is operating,
fixed-point operations can usually be issued to other
units in its pipeline.  In our benchmark, dependencies between
successive division instructions  prevent the processor from using more than one divider.  
Though we have not seen published data on the actual latency and throughput of division and multiplication on this processor, we did not observe the divisor affecting the division instruction's speed, at least within the range of 3 to 4096.

Our results suggest that the gap between multiplication and division performance on the POWER8 lies between that of ARM and Intel; the fact that our approach (using 64-bit multiplications) outperforms the compiler's approach (using 32-bit multiplications) seems to indicate that, unlike ARM, the instruction to compute the most significant bits of a 64-bit product is not much slower than the
instruction to compute a 32-bit product.

Looking at Fig.~\trimmedref{fig:hashbenches-large-d-arm-ryzen}, we see how the approaches compare for 
larger divisors.  The division instruction is sometimes the fastest approach on ARM, and sometimes it
can be faster than the compiler approach on Ryzen.
Overall, our approach is preferred on Ryzen (as well as Skylake and POWER8), but not on ARM\@.

\begin{figure}\centering
\subfloat[Ryzen (GCC)] {
\includegraphics[width=0.49\columnwidth]{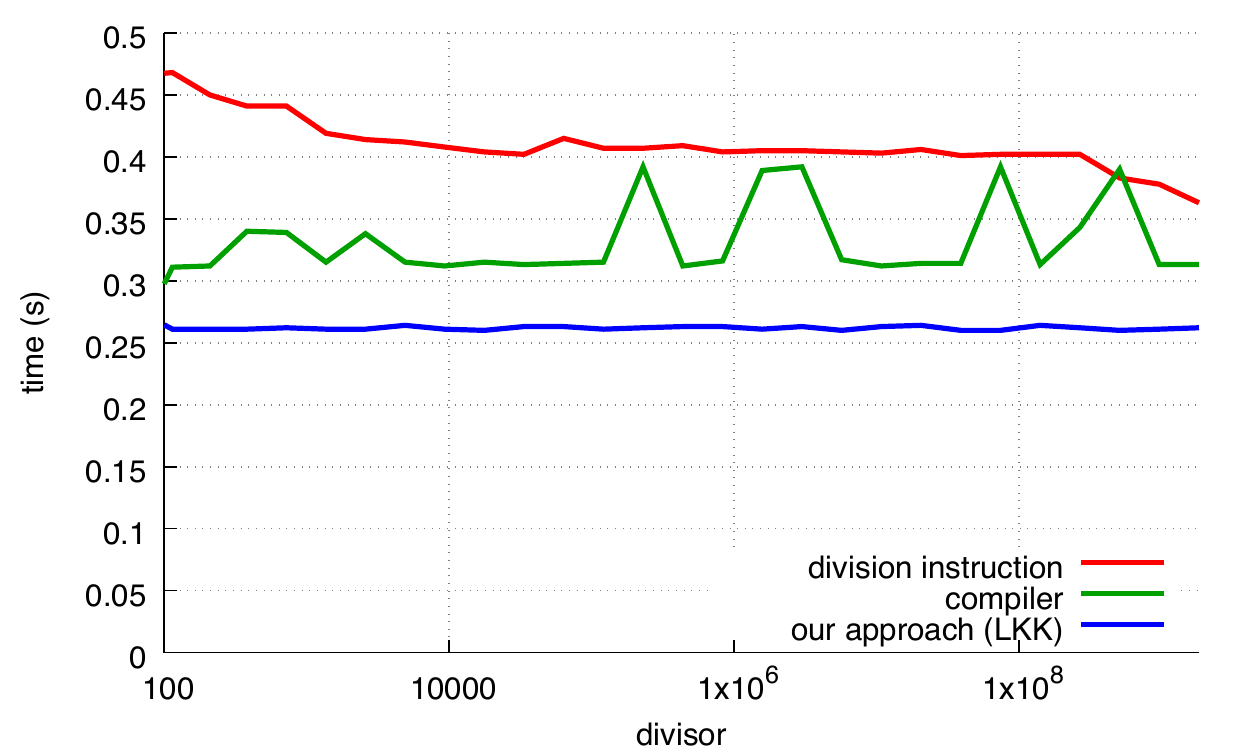}
}\subfloat[ARM (Clang)\label{fig:hashbenches-large-d-arm}] {
\includegraphics[width=0.49\columnwidth]{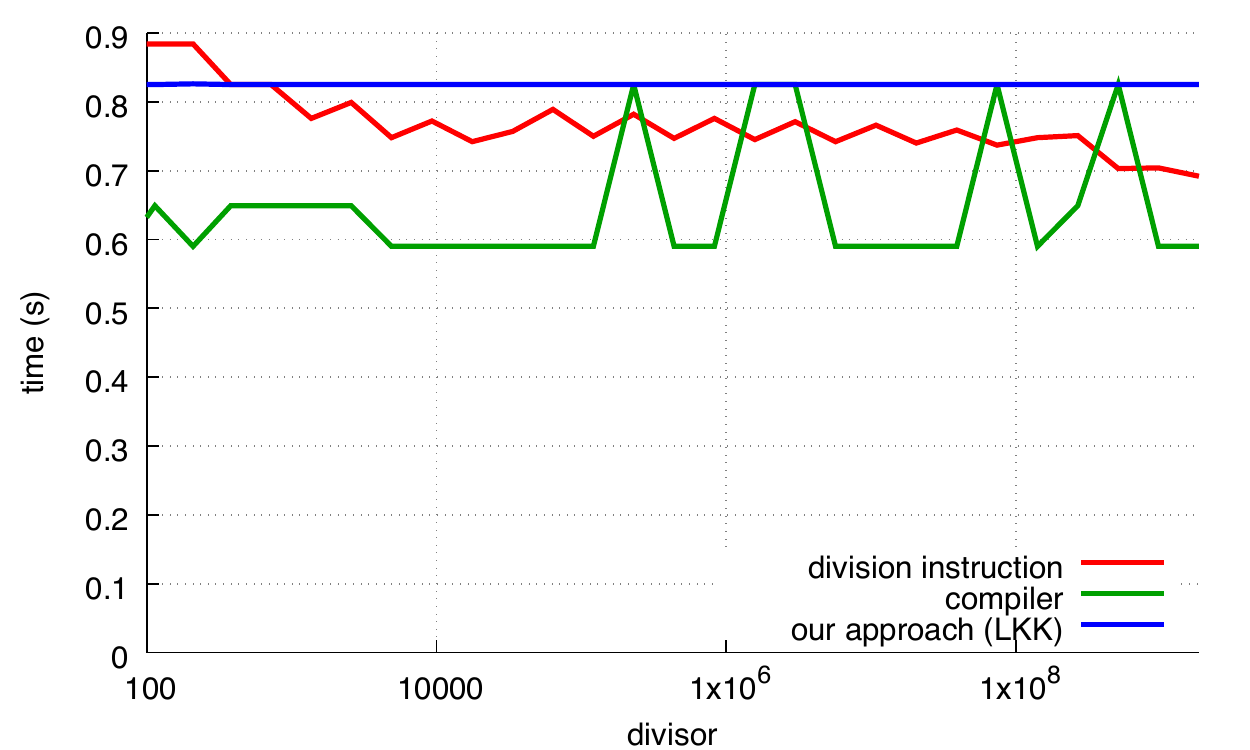}
}
\caption{\label{fig:hashbenches-large-d-arm-ryzen} Ryzen and ARM results for 28 larger divisors (using unsigned arithmetic).
Our approach performed slightly worse when compiled by GCC on ARM, but the Ryzen results were not sensitive to the choice of the compiler.
On Skylake (not shown), the division instruction behaved similarly for large and small divisors, as did compiler and our approach.
}
\end{figure}

\subsection{Beating the libdivide Library}

There are instances when the divisor might not be known at compile time.
 In such instances, we might use a library such as a libdivide.
We once again use our benchmark based on a 
linear congruential generator using the algorithms, but this time, we provide the divisor as a program parameter.

The libdivide library does not have functions to compute the
remainder, so we use its functions to compute the quotient. It has two
types of functions: regular "branchful" ones, those that include some branches that
depend on the divisor, and branchless ones. In this benchmark, the
invariant divisor makes the branches perfectly predictable, and
thus the libdivide branchless functions were always slower.  Consequently we
omit the branchless results.

We present our results in  Fig.~\trimmedref{fig:beatinglibdivide}.
The performance levels of our  functions\footnote{
When 20 test runs were made for 9 divisors, timing results among the 20 never differed by more
than 1\%.
} are insensitive to the divisor, and our performance levels
are always superior to those of the libdivide functions (by about 15\%), except for powers of two in the unsigned case.
In these cases, libdivide is faster, but this is explained by  a fast conditional code path for powers of two.  

\begin{figure}\centering
\subfloat[GNU GCC (unsigned) \label{fig:beatinglibdividegcc}]{
\includegraphics[width=0.49\columnwidth]{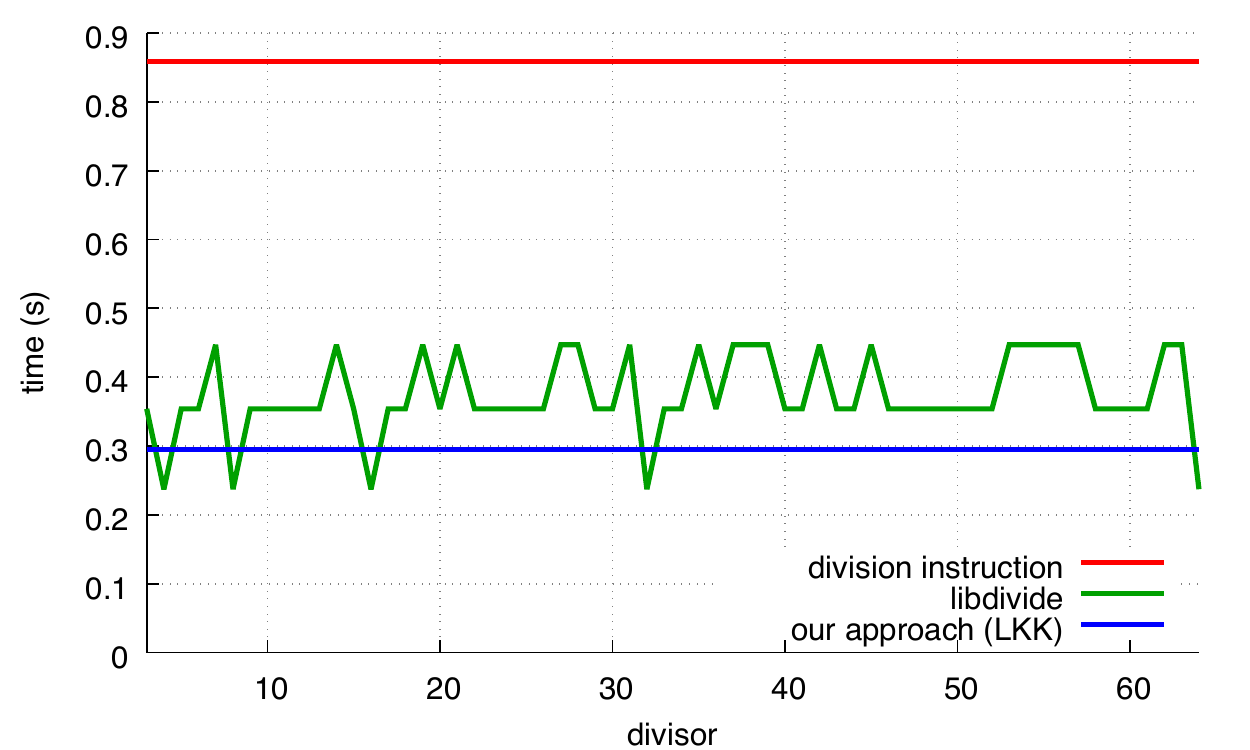}
}
\subfloat[GNU GCC (signed) \label{fig:beatinglibdividegcc-signed}]{
\includegraphics[width=0.49\columnwidth]{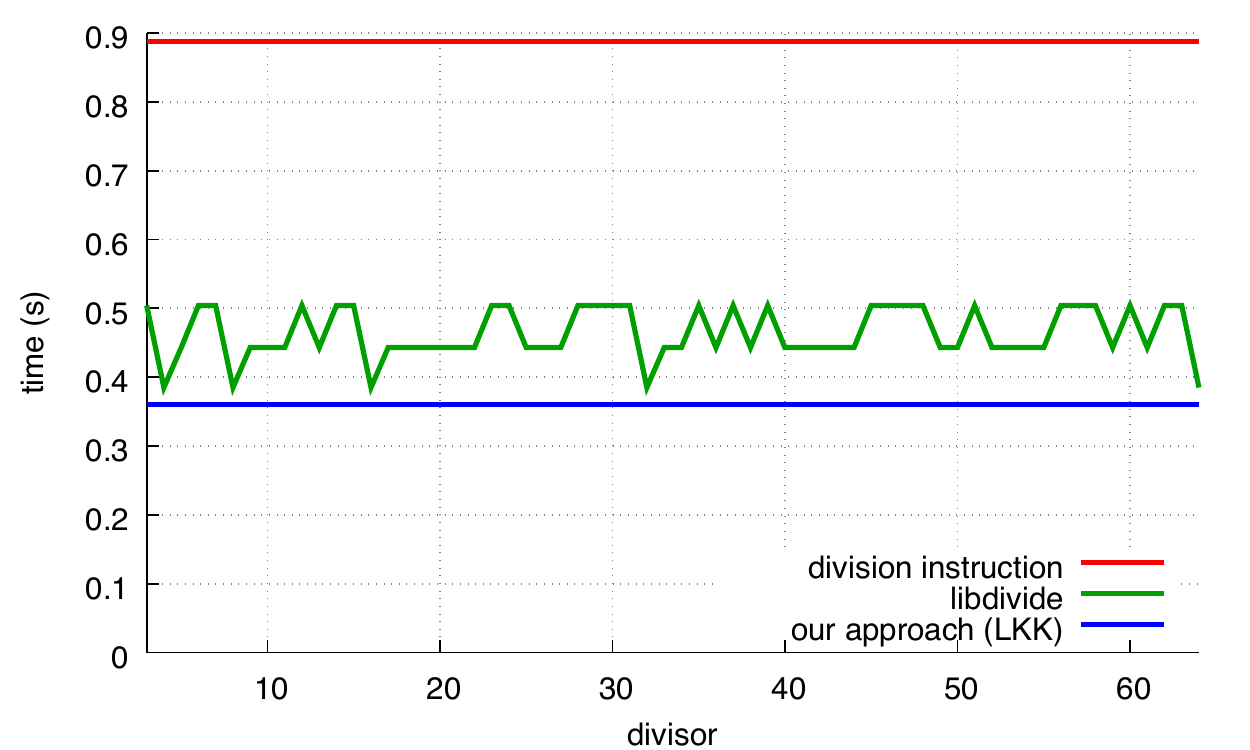}
}
\\
\subfloat[LLVM's Clang (unsigned) \label{fig:beatinglibdivideclang}]{
\includegraphics[width=0.49\columnwidth]{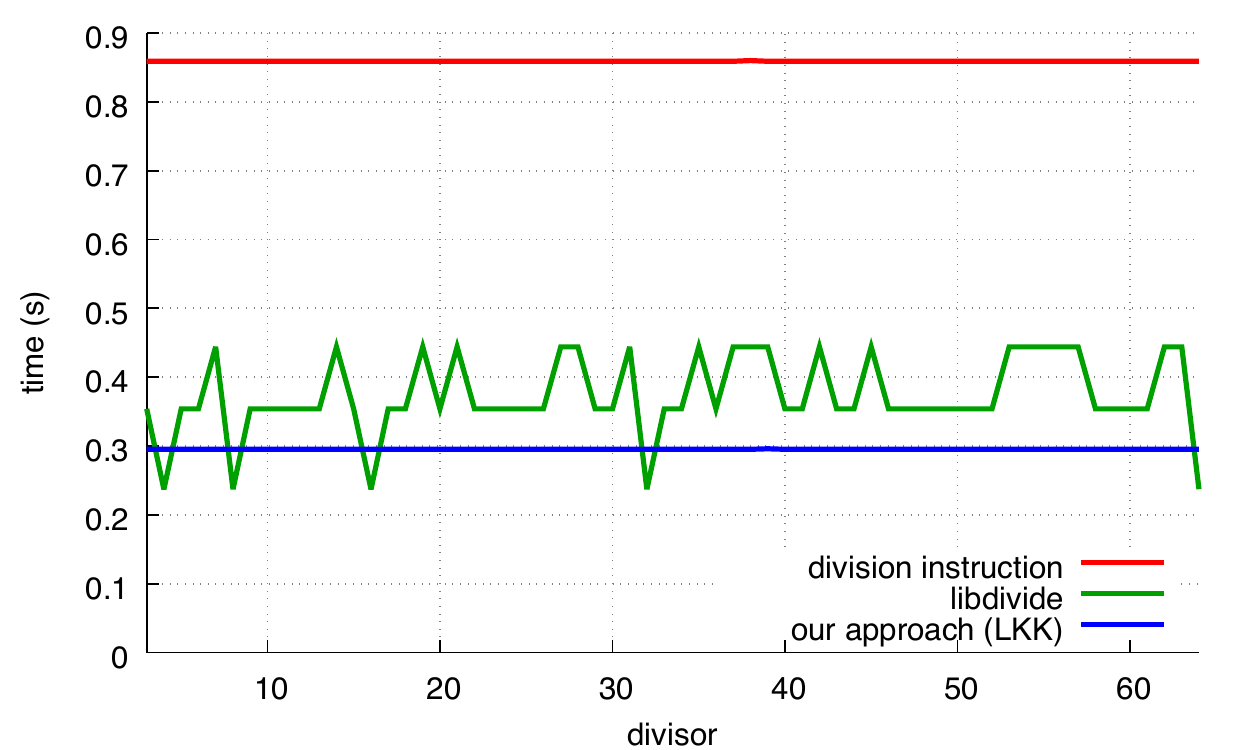}
}
\subfloat[LLVM's Clang (signed) \label{fig:beatinglibdivideclang-signed}]{
\includegraphics[width=0.49\columnwidth]{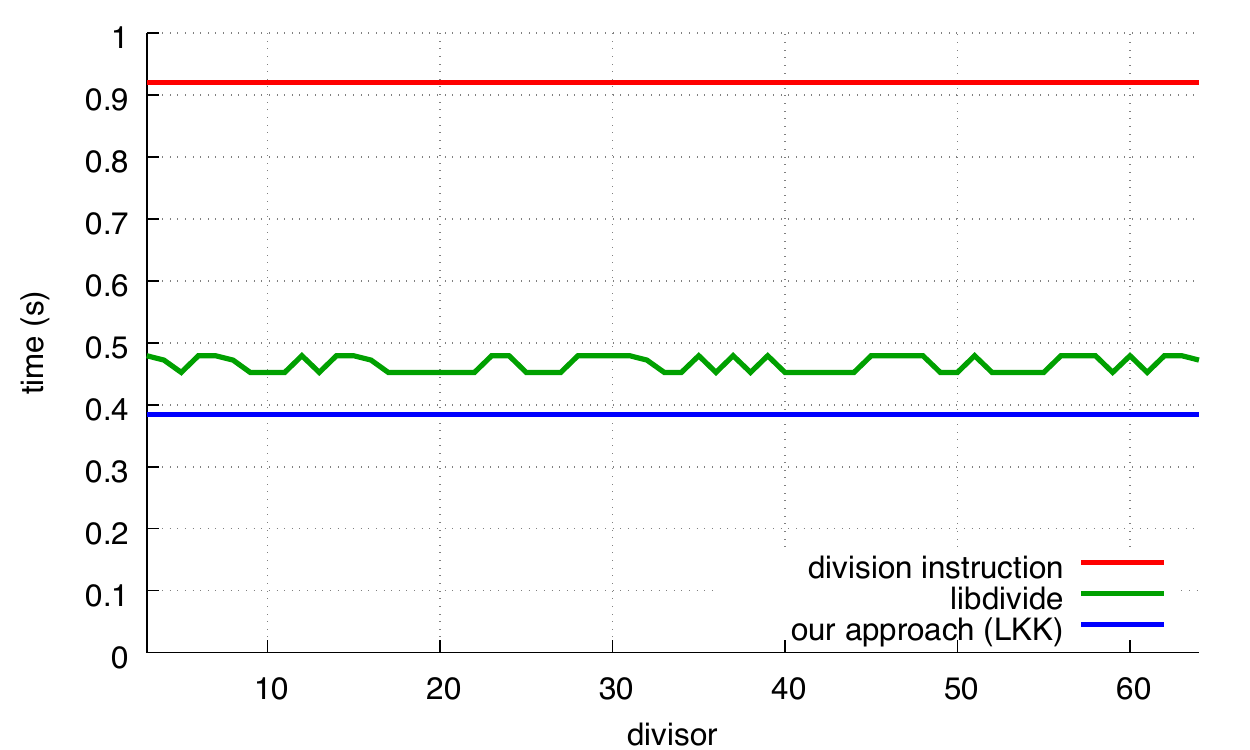}
}
\caption{\label{fig:beatinglibdivide}Wall-clock time to compute 100~million random integers using a linear congruential generator with various divisors passed as a program parameter (Skylake x64).}
\end{figure}

\subsection{Competing for Divisibility}
\label{sec:divisibility-experiments}

We adapted a prime-counting benchmark distributed with libdivide, specialized to
32-bit operands. The code determines the number of primes in
$[2,40000)$ using a simplistic approach: odd numbers in this range
are checked for divisibility by any smaller number that has already been determined
to be prime.  See Fig.~\trimmedref{fig:divisibility-lkk-in-c}.
When a number is identified as a prime,
we compute its scaled approximate reciprocal ($\reciprocal$) value, which is repeatedly used in future
iterations. In this manner, the computation of $\reciprocal$ is only done once per prime, and not
once per trial division by the prime.
A major difference from the benchmark using the linear-congruential generator is that
we cycle rapidly between different divisors, making it much more difficult to
predict branches in the libdivide functions.

\begin{figure}\centering
\begin{minipage}{0.6\textwidth}
\begin{lstlisting}
int count_primes_under_N() {
  int primectr=0;
  static uint64_t prime_cvals[N];

  for (uint32_t n=3; n < N; n += 2) {
    bool isprime=true;
    for (int j=0; j < primectr; ++j) {
      if (lkk_divisible(n, prime_cvals[j])) {
	isprime = false;
	break;
      }
    }
    if (isprime)
      prime_cvals[primectr++] = lkk_cvalue(n);
  }
  return (1+primectr);  // 2 is also prime.
}
\end{lstlisting}
\end{minipage}
\caption{\label{fig:divisibility-lkk-in-c} Prime-counting benchmark for the
   unsigned divisibility test. The code shown is for the LKK approach, similar code is used for other strategies.}
\end{figure}

In these tests, we compare libdivide against LKK and GM, the fast divisibility tests
whose implementations are shown in \S~\ref{sec:divisibility-impl}; see Fig.~\trimmedref{fig:lkk-unsigned-impl} for LKK and Fig.~\trimmedref{fig:gm-unsigned-impl} for GM\@.
Divisibility
of a candidate prime is checked either using
\begin{itemize}
\item libdivide to divide, followed by multiplication and subtraction to determine whether the remainder is nonzero; 
\item the Granlund-Montgomery (GM) divisibility check,
as in Fig.~\trimmedref{fig:gm-unsigned-impl};
\item the C \texttt{\%} operation, which uses a division instruction;
\item our LKK divisibility check (Fig.~\trimmedref{fig:lkk-unsigned-impl}).
\end{itemize}
LKK stores 64 bits for each
prime; GM requires an additional 5-bit rotation amount.
The division-instruction version of the benchmark only needs to store 32~bits per prime.  The libdivide approach requires 72~bits per prime, because we explicitly store the primes.

Instruction counts and execution speed are both important.
All else being equal, we would prefer that compilers
emit smaller instruction sequences.
Using a hardware integer division will yield the smallest
code, but this might give up too much speed.
In the unsigned case, our LKK has a
significant code-size advantage over GM---approximately 3~arithmetic instructions to compute our $\reciprocal$ versus about 11 to compute their required constant.
Both fast approaches use a multiplication and comparison for each subsequent
divisibility check. GM requires an additional instruction to rotate the result
of the multiplication.

Performance results for the unsigned case are shown in
Table~\trimmedref{tbl:libdiv-primes-bench}, showing the total number of
processor cycles on each platform from 1000~repetitions of the benchmark.
On Skylake, 20~repeated computations yielded cycle-count results within 0.3\% of each other.
For ARM, results were always within 4\%. Initially, Ryzen results would sometimes
differ by up to 10\% within 20~attempts, even after we attempted to control such factors as
dynamic frequency scaling.  Thus, rather than reporting the first measurement for each
benchmark, the given Ryzen results are the average of 11~consecutive attempts (the basic benchmark
was essentially executed 11\,000~times).   Our POWER8 results (except one outlier) were within 7\% of one
another over multiple trials and so we averaged several attempts (3 for GCC
and 7 for Clang) to obtain each data point.  Due to
platform constraints, POWER8 results are user-CPU times that matched the wall-clock times.

LKK has a
clear speed advantage in all cases, including the POWER8
and ARM platforms. LKK is between 15\% to 80\% faster than GM\@. Both GM and LKK always are much
faster than using an integer division instruction (up to $7\times$ for Ryzen) 
and they also outperform the best algorithm in libdivide.

\begin{table}
\caption{\label{tbl:libdiv-primes-bench}
Processor cycles (in gigacycles) to determine the number of primes
less than 40000, 1000 times, using unsigned 32-bit computations.  \emph{Branchful} and \emph{branchless} are
libdivide alternatives. Note that libdivide was only available for the x64 systems as it uses platform-specific optimizations.  POWER8 results are in user CPU seconds.
Boldfacing indicates
the fastest approach.
}

\begin{center}
\begin{tabular}{ccccccccc|cc} \toprule
\multirow{2}{*}{Algorithm}          & \multicolumn{2}{c}{Skylake}  & \multicolumn{2}{c}{Haswell} & \multicolumn{2}{c}{Ryzen}   & \multicolumn{2}{c}{ARM} & \multicolumn{2}{c}{POWER8}\\
                   &  GCC    & Clang              &  GCC           &  Clang    & GCC          &  Clang        & GCC         & Clang  & GCC     &  Clang \\ \midrule
            division instruction &  72     & 72                 &  107           &   107     &  131         &   131         &  65        & 65  &  18 & 17\\
        branchful &  46     & 88                 &  56            &   98      &  59          &   71          &  --        & --  & --  & --\\
        branchless &  35     & 35                 &  36            &   37      &  34          &   37          &  --        & -- & -- & --\\
        LKK&\textbf{18} &\textbf{18}      &  \textbf{18}   &\textbf{18}& \textbf{17}  &  \textbf{18}  &   \textbf{27}& \textbf{27} & \textbf{8.7} & \textbf{8.0} \\
        GM &  24     & 27                 &  27            &   28      &  27          &  32           &   36       & 37  & 10 & 11 
        \\\midrule
        GM/LKK & 1.33 & 1.50 & 1.50 & 1.55 & 1.59 & 1.77 & 1.33 & 1.37 & 1.15 & 1.38   
        \\\bottomrule
\end{tabular}
\end{center}
\end{table}



\section{Conclusion}
To our knowledge, we present the first general-purpose algorithms to compute the remainder of the division by unsigned or signed constant divisors directly, using the fractional portion of the product of the numerator with the approximate reciprocal\cite{jacobsohn1973combinatoric,Vowels:1992:D}. On popular x64 processors (and to a lesser extent on POWER), we can produce code for the remainder of the integer division that is faster than
the code produced by well regarded
compilers (GNU GCC and LLVM's Clang) when the divisor constant is known at compile time, using small C functions. Our functions are up to 30\% faster and with only half the number of instructions for most divisors. Similarly, when the divisor is reused, but is not a compile-time constant, we can surpass a popular library (libdivide) by about 15\% for most divisors. 

We can also speed up a test for divisibility. Our approach (LKK) is several times faster than the code produced by popular compilers. It is faster than  the Granlund-Montgomery  divisibility check\cite{Granlund:1994:DII:773473.178249}, sometimes nearly twice as fast. It is advantageous on all tested platforms (x64, POWER8 and 64-bit ARM).

Though compilers already produce efficient code, we show that  additional gains are possible.
As future work, we could compare against more compilers and other libraries. 
Moreover, various additional optimizations
are possible, such as for division by powers of two.

\section*{Acknowledgments}
The work is supported by the
Natural Sciences and Engineering Research Council of Canada
under grant RGPIN-2017-03910.
The authors are grateful to IBM's Centre for Advanced Studies --- Atlantic
and Kenneth Kent for access to the POWER~8
system.


\bibliography{wordmodulo}

\end{document}